\documentclass[10pt]{article}
\usepackage{amsfonts}
\usepackage{amsfonts}
\usepackage{amsfonts}
\usepackage{amsfonts}
\usepackage{amsfonts}
\usepackage{amsfonts}
\usepackage{amsfonts}
\usepackage{amsfonts}
\usepackage{amsfonts}
\usepackage{amsfonts}
\usepackage{amsfonts}
\usepackage{amsfonts}
\usepackage{amsfonts}
\usepackage{amsfonts}
\usepackage{amsfonts}
\usepackage{amsfonts}
\usepackage{amsfonts}
\usepackage{amsfonts}
\usepackage{amsfonts}
\usepackage{mathrsfs}
\usepackage{amsfonts}
\usepackage{amsfonts}
\usepackage{amsfonts}
\usepackage{amsfonts}
\usepackage{amsfonts}
\usepackage{amsfonts}
\usepackage{amsfonts}
\usepackage{amsfonts}
\usepackage{amsfonts}
\usepackage{amsfonts}
\usepackage{amsfonts}
\usepackage{amsfonts}
\usepackage{amsfonts}
\usepackage{amsfonts}
\usepackage{amsfonts}
\usepackage{amsfonts}
\usepackage{amsfonts}
\usepackage{amsfonts}
\usepackage{amsfonts}
\usepackage{amsfonts}
\usepackage{amsfonts}
\usepackage{amsfonts}
\usepackage{amsfonts}
\usepackage{amsfonts}
\usepackage{amsfonts}
\usepackage{}
\usepackage{amsfonts}
\usepackage{amsmath}
\usepackage{amssymb, amscd, amsthm}
\usepackage[all]{xy}
\usepackage[dvips]{graphicx}
\usepackage{verbatim}
\usepackage{cite}
\usepackage{enumerate}
\usepackage{float}
\usepackage{diagbox}
\numberwithin{equation}{section}

\newtheorem{lemma}{\textbf{Lemma}}[section]
\newtheorem{theorem}{\textbf{Theorem}}
\newtheorem{remark}{\textbf{Remark}}[section]
\newtheorem{corollary}{\textbf{Corollary}}[section]

\newtheorem{example}{\textbf{Example}}[section]
\newtheorem{proposition}{\textbf{Proposition}}[section]

\newtheorem{Definition}{\textbf{Definition}}
\usepackage{longtable}

\usepackage{multirow}

\begin{document}
\baselineskip 17pt\title{\Large\bf Block Codes in Pomset Metric over $\mathbb{Z}_m$}
\author{\large  Wen Ma \quad\quad Jinquan Luo\footnote{The authors are with School of Mathematics
and Statistics \& Hubei Key Laboratory of Mathematical Sciences, Central China Normal University, Wuhan China.\newline
 E-mails: mawen95@126.com(W.Ma),  luojinquan@mail.ccnu.edu.cn(J.Luo)}}
\date{}
\maketitle

{\bf Abstract}: In this paper, we introduce codes equipped with pomset block metric. A Singleton type bound for pomset block codes is obtained. Code achieving the Singleton bound, called a maximum distance separable code (for short, MDS ($\mathbb{P},\pi$)-code) is also investigated. We extend the concept of $I$-perfect codes and $r$-perfect codes to pomset block metric. The relation between $I$-perfect codes and MDS $(\mathbb{P},\pi)$-codes is also considered. When all blocks have the same dimension, we prove the duality theorem for codes and study the weight distribution of MDS pomset block codes when the pomset is a chain.

{\bf Key words}: pomset, label, block, MDS codes, perfect codes.

\section{Introduction}

 \quad\; The codes equipped with a metric differ from the Hamming metric has been studied for years. Poset metric was introduced by Brualdi (see [\ref{POSET}]) in 1995. The concept of poset metric was motivated by Niederreiter's generalization of a classical problem on coding theory (see [\ref{NIEDE1}] and [\ref{NIEDE2}]).
 In [\ref{POSET2008}], Hyun and Kim introduced the concept of $I$-perfect codes and described the MDS poset codes in terms of $I$-perfect codes. They also studied the weight distribution of an MDS poset code and proved the duality theorem. Feng, Xu, and Hickernell (see [\ref{FENG}]) introduced the block metric, by partitioning the set of coordinate positions of $\mathbb{F}_q^n$ into families of blocks. Later, Alves, Panek and  Firer introduced poset block metric (see [\ref{Alves}]) and studied NRT block metric (see [\ref{PANEK}]). By extending their observations, Dass, Sharma and Verma studied poset block codes and defined a maximum distance separable poset block code. Moreover, they extended the concept of $I$-balls to poset block metric and described $r$-perfect and MDS $(\mathbb{P},\pi)$-codes (see [\ref{DASS}]). When all the blocks have the same dimension, they showed that MDS $(\mathbb{P},\pi)$-codes are the same as $I$-perfect codes for some ideals $I$.

Recently, Irrinki and Selvaraj (see [\ref{POMSET1}]) introduced pomset metric and enhanced the concept of order ideals. Construction of pomset codes are obtained and their metric properties like minimum distance and covering radius are determined. In [\ref{POMSET2}], Irrinki and Selvaraj studied $I$-perfect codes under pomset metric by extending the concept of $I$-balls in poset metric. Moreover, they established a Singleton type bound for codes with pomset metric and investigated the connection between MDS codes and $I$-perfect codes. When the pomset is a chain, they proved the duality theorem and determined the weight distribution of MDS pomset codes.

A linear error-block code is a natural generalization of the classical error-correcting code and has applications in experimental design, high-dimensional numerical integration and cryptography. The construction of linear error-block codes with the largest rate, $k/n$, and the minimum distance $d$ is an important problem in coding theory. The support of $x\in\mathbb{F}_q^n$ given by $supp(x)=\{i:x_i\neq 0\}$ is a set. The poset weight $w_p(x)$ of $x$ is defined as $w_p(x)=|<supp(x)>|$ and $d_p(x,y)=w_p(x-y)$ is a well defined metric on $\mathbb{F}_q^n$. One can get different metric on $\mathbb{F}_q^n$ by varying posets such as Rosenbloom-Tsfasman (RT)-metric if $P$ is a chain, Hamming metric if $P$ is an antichain and so on.  In 2008, Firer, Paneck  and Alves ([\ref{Alves}]) presented the family of metrics called poset-block metric that generalizes all the previous ones. But the poset metric does not accodomate Lee metric for any particular poset. For an element $l\in\mathbb{Z}_m$, the Lee weight of $l$ is defined as $w_L(l)=\min\{l,m-l\}$ whereas the Hamming weight of any $l\neq 0$ is 1. Moreover, the Hamming weight of $x=(x_1,x_2,\ldots,x_n)\in\mathbb{Z}_m^n$ is sum of the Hamming weights of the non-zero coordinates, so that it counts the number of non-zero positions whereas Lee weight adds Lee weight of  non-zero coordinates in $x$. Thus, the support of $x\in\mathbb{Z}_m^n$ with respect to Lee weight is to be defined as $supp_L(x)=\{k/i:k=w_L(x_i),k\neq0\}$ which is a multiset.  The pomset metric is a generalization to Lee metric when the pomset is an antichain. In this paper, we combine the pomset and block structure to obtain a further generalization called the pomset block metric. Pomset block codes reduces to the pomset codes with $\pi=[1]^n$ and thus reduces to Lee metric with chain pomset. In some sense, it is a generalization to poset-block metric. By researching pomset-block codes, we give a much general method to handle Lee metric and error-block codes over $\mathbb{Z}_m$. This paper aims to introduce pomset block codes and extend the concept of $I$-perfect codes to the case of pomset block metric. A Singleton type bound for pomset block codes is established and the relationship between MDS codes and $I$-perfect codes is investigated. We also prove the duality result for $I$-perfect code when all the blocks have the same dimension. When the pomset is a chain and all the blocks have the same dimension, we determine the weight distribution of an MDS $(\mathbb{P},\pi)$-code.

\section{Preliminaries}

\quad\; In this section, we introduce some basic notations and useful results of a pomset block metric.

A collection of elements which may contain duplicates is called a \textbf{multiset} (in short, \textbf{mset}). Girish and John defined a multiset relation and explored some of basic properties (see [\ref{GIRISH}] and [\ref{girish}]). They also defined a partially ordered multiset as a multiset relation being reflexive, antisymmetric and transitive, chains and antichains  of a partially ordered multiset.

Formally, if $X$ is a set of elements, a mset $M$ drawn from the set $X$ is represented by a function count $C_{M}:X\rightarrow\mathbb{Z}_{+}$ where $\mathbb{Z}_{+}$ represents the set of non-negative integers. For each $a\in X$, $C_{M}(a)$ indicates the number of occurrences of the element $a$ in $M$.

An element $a\in X$ appearing $p$ times in $M$ is denoted by $p/a\in M$ and thus $C_M(a)=p$. If we consider $k/a\in M$, the value of $k$ satisfies $k\leq p$. The mset drawn from the set $X=\{a_1,a_2,\ldots,a_n\}$ is represented as $M=\{p_1/a_1,p_2/a_2,\ldots,p_n/a_n\}$. An mset is called \textbf{regular} if all its objects occur with the same multiplicity and the common multiplicity is called its \textbf{height}. The \textbf{cardinality} of an mset $M$ drawn from $X$ is $|M|=\sum_{a\in X}C_M(a)$. The \textbf{root set} of $M$ denoted by $M^{*}$ is defined as $M^{*}=\{a\in X:C_M(a)>0\}$.

Let $M_1$ and $M_2$ be two msets drawn from a set $X$. We call $M_1$ a \textbf{submset} of $M_2$ ($M_1\subseteq M_2$) if $C_{M_1}(a)\leq C_{M_2}(a)$ for all $a\in X$. We call $M_1$ a \textbf{proper submset} of $M_2$ ($M_1\subset M_2)$ if there exists at least one $a\in M_2^*$ such that $C_{M_1}(a)<C_{M_2}(a)$. Two msets $M_1$ and $M_2$ are equal ($M_1=M_2$) if $M_1\subseteq M_2$ and $M_2\subseteq M_1$.

Let $M'$ be a submset of $M$. An element $a\in M'^{*}$ is said to have \textbf{full count with respect to $M$} if $C_{M'}(a)=C_{M}(a)$. $M'$ is said to have full count if for all $a\in M'^{*}$, one has $C_{M'}(a)=C_{M}(a)$; otherwise, $M'$ is said to be with \textbf{partial count}.

Let $M_1$ and $M_2$ be two msets drawn from a set $X$. \textbf{Addition (sum)} of $M_1$ and $M_2$ denoted by $M=M_1\oplus M_2$ is defined as $C_{M}(a)=C_{M_1}(a)+C_{M_2}(a)$ for all $a\in X$. \textbf{Subtraction (difference)} of $M_2$ from $M_1$ denoted by $M=M_1\ominus M_2$ is defined as $C_M(a)=\max\left\{C_{M_1}(a)-C_{M_2}(a),0\right\}$ for all $a\in X$.
The \textbf{union} of $M_1$ and $M_2$ is an mset denoted by $M=M_1\cup M_2$ such that for all $a\in X$, $C_M(a)=\text{max}\{C_{M_1}(a),C_{M_2}(a)\}$.

The \textbf{mset space} $[X]^{l}$ is the set of all msets drawn from $X$ such that no element in an mset occurs more than $l$ times. If $M_1,M_2\in[X]^l$, the \textbf{mset sum} would be modified as
\begin{center}
$C_{M_1\oplus M_2}(a)=\text{min}\left\{l,C_{M_1}(a)+C_{M_2}(a)\right\}$ for all $a\in X$.
\end{center}
Let $M\in[X]^l$ be an mset, the \textbf{complement} $M^{c}$ is an element of $[X]^l$ such that $C_{M^{c}}(a)=l-C_{M}(a)$ for all $a\in X$.

Let $M_1$ and $M_2$ be two msets drawn from $X$, the \textbf{Cartesian product} of $M_1$ and $M_2$ is also an mset defined as
$$M_1\times M_2=\{pq/(p/a,q/b):p/a\in M_1,q/b\in M_2\}.$$
Denote by $C_1(a,b)$ the count of the first coordinate in the ordered pair $(a,b)$ and by $C_2(a,b)$ the count of the second coordinate in the ordered pair $(a,b)$.

A submset $R$ of $M\times M$ is said to be an \textbf{mset relation} on $M$ if every member $(p/a,q/b)$ of $R$ has count $C_1(a,b)\cdot C_2(a,b)$. An mset relation $R$ on an mset $M$ is said to be \textbf{reflexive} if $m/a\ R\ m/a$ for all $m/a\in M$; \textbf{antisymmetric} if $m/a\ R\ n/b$ and $n/b\ R\ m/a$ imply $m=n$ and $a=b$; \textbf{transitive} if $m/a\ R\ n/b$ and $n/b\ R\ k/c$ imply $m/a\ R\ k/c$. An mset relation $R$ is called a \textbf{partially ordered mset relation (or pomset relation)} if it is reflexive, antisymmetric and transitive. The pair $(M,R)$ is known as a \textbf{partially ordered multiset (pomset)} denoted by $\mathbb{P}$.

Let $\mathbb{P}=(M,R)$ and $m/a\in M$. Then $m/a$ is a \textbf{maximal element} of $\mathbb{P}$ if there exists no $n/b\in M\ (b\neq a)$ such that $m/a\ R\ n/b$; $m/a$ is a \textbf{minimal element} if there exists no $n/b\in M\ (b\neq a)$ such that $n/b\ R\ m/a$. $\mathbb{P}$ is called a \textbf{chain} if every distinct pair of points from $M$ is comparable in $\mathbb{P}$. $\mathbb{P}$ is called an \textbf{antichain} if $n/b\ R\ m/a$ implies $a=b$.

A submset $I$ of $M$ is called an \textbf{order ideal} (or simply an \textbf{ideal}) of $\mathbb{P}$ if $k/a\in I$ and $q/b\ R\ k/a$ $(b\neq a)$ imply $q/b\in I$.  An \textbf{ideal generated by an element $k/a\in M$} is defined as
$$\langle k/a\rangle=\{k/a\}\cup\{q/b\in M: q/b\ R\ k/a\  \text{and}\ b\neq a\}. $$

An \textbf{ideal generated by a submset} $S$ of $M$ is defined by $\langle S\rangle=\bigcup\limits_{k/a\in S}\langle k/a\rangle$. We use $\mathcal {I}(\mathbb{P})$ (resp. $\mathcal {I}^{r}(\mathbb{P}))$ to denote the set of all ideals of $\mathbb{P}$ (resp. of cardinality $r$). An ideal $I$ is said to be \textbf{an ideal with full count} if $C_{I}(a)=C_{M}(a)$ for all $a\in I^{*}$, otherwise, $I$ is said to be \textbf{an ideal with partial count}. We denote the set of all maximal elements in an ideal $I$ by $M(I)$.

\begin{example}
Let us consider the case $M=\{2/1,2/2,2/3\}$ and $R$ be a $V$-shape poset i.e., $1<2$, $1<3$ and 2, 3 are incomparable. Then the maximal elements of pomset $\mathbb{P}=(M,R)$ are $2/2$ and $2/3$ and the minimal element is $2/1$. The ideals of $\mathbb{P}$ with cardinality $1$ is $I_1=\{1/1\}$, with cardinality 2 is $I_2=\{2/1\}$, with cardinality 3 are $I_3=\{2/1,1/2\}$ and $I_4=\{2/1,1/3\}$, with cardinality 4 are $I_5=\{2/1,2/2\}$, $I_6=\{2/1,2/3\}$ and $I_7=\{2/1,1/2,1/3\}$, with cardinality 5 are $I_8=\{2/1,2/2,1/3\}$ and $I_9=\{2/1,1/2,2/3\}$, with cardinality 6 is $I_{10}=\{2/1,2/2,2/3\}$ respectively.
\end{example}

\begin{proposition}([\ref{POMSET1}])\label{pomset1}
Let $\mathbb{P}=(M,R)$ be a pomset. Then
\begin{enumerate}[(1)]
\item for any $I\in\mathcal {I}^r(\mathbb{P})$ and $0\leq s\leq r\leq |M|$, there exists $J\in\mathcal {I}^s(\mathbb{P})$ such that $J\subseteq I$.
\item for any $I\in\mathcal {I}^r(\mathbb{P})$ and $0\leq r\leq s\leq |M|$, there exists $J\in\mathcal {I}^s(\mathbb{P})$ such that $I\subseteq J$.
\end{enumerate}
\end{proposition}

For a given pomset $\mathbb{P}=(M,R)$, we define the \textbf{dual pomset} $\widetilde{\mathbb{P}}=(M,\widetilde{R})$ as follows:
\begin{center}
$\mathbb{P}$ and $\widetilde{\mathbb{P}}$ have the same underlying set $M$ and $p/a\ R\ q/b$ in $\mathbb{P}$ if and only if $q/b\ \widetilde{R}\ p/a$ in $\widetilde{\mathbb{P}}$.
\end{center}

\begin{proposition}([\ref{POMSET2}])
Let $M\in [X]^l$ be a regular mset with height $l$. Let $\mathbb{P}$ be a pomset on $M$ and $\tilde{\mathbb{P}}$ be its dual pomset. Then the order ideals of $\tilde{\mathbb{P}}$ are precisely the complements of the order ideals of $\mathbb{P}$, that is, $\mathcal {I}(\widetilde{\mathbb{P}})=\left\{I^{c}:I\in \mathcal {I}(\mathbb{P})\right\}$.
\end{proposition}

Consider $\mathbb{Z}_m=\{0,1,\ldots,m-1\}$, the ring of integers modulo $m$ for $m\geq 4$. We consider a pomset $\mathbb{P}$ defined on an mset $M=\left\{\left\lfloor\frac{m}{2}\right\rfloor/1,\left\lfloor\frac{m}{2}\right\rfloor/2, \ldots,\left\lfloor\frac{m}{2}\right\rfloor/s\right\}\in [X]^{\left\lfloor\frac{m}{2}\right\rfloor}$ where $X=[s]$.

Let $\pi: [s]\rightarrow \mathbb{N}$ be a map such that $n=\sum\limits_{i=1}\limits^{s}\pi(i)$. The map $\pi$ is said to be a \textbf{labeling} of the pomset $\mathbb{P}$, and the pair $(\mathbb{P}, \pi)$ is called a \textbf{pomset block structure} over $[s]$. Denote $\pi(i)$ by $k_i$ and take $V_i$ as free $\mathbb{Z}_m$-module $\mathbb{Z}_m^{k_i}$ for all $1\leq i\leq s$. Define $V$ as
$$V=V_1\oplus V_2\oplus\cdots\oplus V_s$$
which is isomorphic to $\mathbb{Z}_m^{n}$. Each $u\in\mathbb{Z}_m^{n}$ can be written as $u=(u_1,u_2,\ldots,u_s)$ where $u_i\in \mathbb{Z}_m^{k_i}$, $1\leq i\leq s$. The \textbf{Lee block support} of $u\in\mathbb{Z}_m^{n}$ is defined as
$$supp_{(L,\pi)}(u)=\left\{s_i/i:s_i=w_{(L,\pi)}(u_i),s_i\neq 0\right\},$$
where
$$w_{(L,\pi)}(u_i)=\max\left\{\min\left\{u_{i_t}, m-u_{i_t}\right\}:1\leq t\leq \pi(i)\right\}.$$
The \textbf{$(\mathbb{P},\pi)$-weight} of $u\in\mathbb{Z}_m^{n}$ is defined to be the cardinality of the ideal generated by $supp_{(L,\pi)}(u)$, that is
$$w_{(\mathbb{P},\pi)}(u)=\left|\langle supp_{(L,\pi)}(u)\rangle\right|.$$
The \textbf{pomset block distance} between two vectors $u,v\in\mathbb{Z}_m^{n}$ is defined as
$$d_{(\mathbb{P},\pi)}(u,v)=w_{(\mathbb{P},\pi)}(u-v).$$

Now we prove that the above pomset block distance is a metric on $\mathbb{Z}_m^{n}$.

\begin{theorem}
 Let $\mathbb{P}$ be a pomset on a regular mset $M=\left\{\left\lfloor\frac{m}{2}\right\rfloor/1,\left\lfloor\frac{m}{2}\right\rfloor/2, \ldots,\left\lfloor\frac{m}{2}\right\rfloor/s\right\}$, and $\pi:[s]\rightarrow \mathbb{N}$ such that $n=\sum\limits_{i=1}\limits^{s}\pi(i)$ be a labeling of $\mathbb{P}$. Then the pomset block distance $d_{(\mathbb{P},\pi)}(.,.)$ is a metric on $\mathbb{Z}_m^{n}$.
\end{theorem}

\begin{proof}
It is obvious that $d_{(\mathbb{P},\pi)}(u,v)\geq 0$ and $d_{(\mathbb{P},\pi)}(u,v)=0$ if and only if $u=v$. Let $u,v,w\in\mathbb{Z}_m^{n}$. Since $supp_{(L,\pi)}(u-v)=supp_{(L,\pi)}(v-u)$, we have $d_{(\mathbb{P},\pi)}(v,u)=d_{(\mathbb{P},\pi)}(u,v)$. It remains to show that $d_{(\mathbb{P},\pi)}(u,v)\leq d_{(\mathbb{P},\pi)}(u,w)+d_{(\mathbb{P},\pi)}(w,v)$.
Since $d_{(\mathbb{P},\pi)}(u,v)=w_{(\mathbb{P},\pi)}(u-v)=w_{(\mathbb{P},\pi)}(u-w+w-v)$, it suffices to show that the $(\mathbb{P},\pi)$-weight satisfies the inequality $w_{(\mathbb{P},\pi)}(x+y)\leq w_{(\mathbb{P},\pi)}(x)+w_{(\mathbb{P},\pi)}(y)$ for all $x,y\in\mathbb{Z}_m^{n}$. We only need to prove that $w_{(L,\pi)}(x_i+y_i)\leq w_{(L,\pi)}(x_i)+w_{(L,\pi)}(y_i)$ for all $i\in [s]$. Note that $$w_{(L,\pi)}(x_i+y_i)=\max\left\{\min\{x_{i_t}+y_{i_t},m-x_{i_t}-y_{i_t}\}: 1\leq t\leq \pi(i)\right\}.$$
Suppose that $w_{(L,\pi)}(x_i+y_i)=\min\{x_{i_1}+y_{i_1}, m-x_{i_1}-y_{i_1}\}$.

\begin{itemize}
  \item {\bf Case 1:} Assume $x_{i_1},y_{i_1}\leq\left\lfloor\frac{m}{2}\right\rfloor$ and $x_{i_1}+y_{i_1}\leq \left\lfloor\frac{m}{2}\right\rfloor$. We have $w_{(L,\pi)}(x_i+y_i)=x_{i_1}+y_{i_1}$. On the other hand, one has $\min\{x_{i_1},m-x_{i_1}\}=x_{i_1}\leq w_{(L,\pi)}(x_i)$ and $\min\{y_{i_1},m-y_{i_1}\}=y_{i_1}\leq w_{(L,\pi)}(y_i)$. Therefore,
      $w_{(L,\pi)}(x_i+y_i)=x_{i_1}+y_{i_1}\leq w_{(L,\pi)}(x_i)+w_{(L,\pi)}(y_i)$.
  \item {\bf Case 2:} Assume $x_{i_1},y_{i_1}\leq\left\lfloor\frac{m}{2}\right\rfloor$ and $x_{i_1}+y_{i_1}> \left\lfloor\frac{m}{2}\right\rfloor$. Then $w_{(L,\pi)}(x_i+y_i)=m-x_{i_1}-y_{i_1}\leq x_{i_1}+y_{i_1}\leq w_{(L,\pi)}(x_i)+w_{(L,\pi)}(y_i)$.
  \item {\bf Case 3:} Assume $x_{i_1},y_{i_1}>\left\lfloor\frac{m}{2}\right\rfloor$ and $x_{i_1}+y_{i_1}\leq\left\lfloor\frac{m}{2}\right\rfloor$. We have $w_{(L,\pi)}(x_i+y_i)=x_{i_1}+y_{i_1}-m\leq\left\lfloor\frac{m}{2}\right\rfloor$. On the other hand, one has $\min\{x_{i_1},m-x_{i_1}\}=m-x_{i_1}\leq w_{(L,\pi)}(x_i)$ and $\min\{y_{i_1},m-y_{i_1}\}=m-y_{i_1}\leq w_{(L,\pi)}(y_i)$. So, $w_{(L,\pi)}(x_i)+w_{(L,\pi)}(y_i)\geq m-x_{i_1}+m-y_{i_1}\geq\left\lfloor\frac{m}{2}\right\rfloor$ and hence $w_{(L,\pi)}(x_i+y_i)\leq\left\lfloor\frac{m}{2}\right\rfloor =w_{(L,\pi)}(x_i)+w_{(L,\pi)}(y_i)$.
  \item {\bf Case 4:} Assume $x_{i_1},y_{i_1}>\left\lfloor\frac{m}{2}\right\rfloor$ and $x_{i_1}+y_{i_1}>\left\lfloor\frac{m}{2}\right\rfloor$. Then $w_{(L,\pi)}(x_i+y_i)=2m-x_{i_1}-y_{i_1}\leq w_{(L,\pi)}(x_i)+w_{(L,\pi)}(y_i)$.
  \item {\bf Case 5:} Assume $x_{i_1}\leq\left\lfloor\frac{m}{2}\right\rfloor$, $y_{i_1}>\left\lfloor\frac{m}{2}\right\rfloor$ and $x_{i_1}+y_{i_1}>\left\lfloor\frac{m}{2}\right\rfloor$. We have $w_{(L,\pi)}(x_i+y_i)=m-x_{i_1}-y_{i_1}$. On the other hand, one has $\min\{x_{i_1},m-x_{i_1}\}=x_{i_1}\leq w_{(L,\pi)}(x_i)$ and $\min\{y_{i_1},m-y_{i_1}\}=m-y_{i_1}\leq w_{(L,\pi)}(y_i)$. So, $w_{(L,\pi)}(x_i)+w_{(L,\pi)}(y_i)\geq x_{i_1}+m-y_{i_1}\geq m-x_{i_1}-y_{i_1}=w_{(L,\pi)}(x_i+y_i)$.
  \item {\bf Case 6:} Assume $x_{i_1}\leq\left\lfloor\frac{m}{2}\right\rfloor$, $y_{i_1}>\left\lfloor\frac{m}{2}\right\rfloor$ and $x_{i_1}+y_{i_1}\leq\left\lfloor\frac{m}{2}\right\rfloor$. Then $w_{(L,\pi)}(x_i+y_i)=x_{i_1}+y_{i_1}-m\leq x_{i_1}+m-y_{i_1}\leq w_{(L,\pi)}(x_i)+w_{(L,\pi)}(y_i)$.
\end{itemize}
This completes our proof.
\end{proof}

The metric $d_{(\mathbb{P},\pi)}(.,.)$ on $\mathbb{Z}_m^{n}$ is called a \textbf{pomset block metric}. The pair $(\mathbb{Z}_m^n, d_{(\mathbb{P},\pi)})$ is said to be a \textbf{pomset block space}. A subset $\mathcal {C}$ of $(\mathbb{Z}_m^n,d_{(\mathbb{P},\pi)})$ with cardinality $K$ is called an $(n,K,d)$ \textbf{$(\mathbb{P},\pi)$-code}, where $\mathbb{Z}_m^n$ is equipped with the pomset block metric $d_{(\mathbb{P},\pi)}(.,.)$ and
$$d=d_{(\mathbb{P},\pi)}(\mathcal {C})=\min\left\{w_{(\mathbb{P},\pi)}(c-c'):c, c'\in \mathcal {C}\right\}$$
is the $(\mathbb{P},\pi)$-minimum distance of $\mathcal {C}$. If $\mathcal {C}$ is a submodule of $\mathbb{Z}_m^n$ with cardinality $m^k$,  we call $\mathcal {C}$ a \textbf{linear $(n,m^k,d)$ $(\mathbb{P},\pi)$-code}.
The \textbf{dual} of an $(n,K,d)$ $(\mathbb{P},\pi)$-code $\mathcal {C}$ is defined as
$$\mathcal {C}^{\bot}=\left\{v\in\mathbb{Z}_{m}^{n}:c\cdot v=c_1v_1+\cdots+c_nv_n=0\ \text{for all}\ c\in \mathcal {C}\right\}.$$

\section{Pomset balls and MDS $(\mathbb{P},\pi)$-codes}

\subsection{$r$-balls}

\quad\; Let $\mathbb{P}$ be a pomset on $M=\left\{\left\lfloor\frac{m}{2}\right\rfloor/1,\left\lfloor\frac{m}{2}\right\rfloor/2, \ldots,\left\lfloor\frac{m}{2}\right\rfloor/s\right\} $ and $\pi$ be a labeling on $\mathbb{P}$. For $u\in\mathbb{Z}_m^n$ and a non-negative integer $r$, the \textbf{$(\mathbb{P}, \pi)$-ball} with center $u$ and radius $r$ is the set
$$B_r(u)=\left\{v\in\mathbb{Z}_m^n:d_{(\mathbb{P},\pi)}(u,v)\leq r\right\}.$$
The \textbf{$(\mathbb{P},\pi)$-sphere} with center $u\in\mathbb{Z}_m^n$ and radius $r$ denoted by $S_r(u)$ is defined as the set of all those vectors of $\mathbb{Z}_m^n$ whose pomset block distance is equal to $r$, that is,
$$S_r(u)=\left\{v\in\mathbb{Z}_m^n:d_{(\mathbb{P},\pi)}(u,v)=r\right\}.$$
For simplify, $B_r(u)$ and $S_r(u)$ are also called \textbf{$r$-ball} and \textbf{$r$-sphere} centered at $u$ respectively.

\begin{Definition}
A $(\mathbb{P},\pi)$-code $\mathcal {C}$ is said to be an $r$-error correcting $(\mathbb{P},\pi)$-code if the $(\mathbb{P},\pi)$-balls centered at the codewords of $\mathcal {C}$ are pairwise disjoint.
\end{Definition}

\begin{Definition}
A $(\mathbb{P},\pi)$-code $\mathcal {C}$ is said to be an $r$-perfect $(\mathbb{P},\pi)$-code if the $(\mathbb{P},\pi)$-balls of radius $r$ centered at the codewords of $\mathcal {C}$ are pairwise disjoint and their union covers the entire space $\mathbb{Z}_m^n$.
\end{Definition}

We now consider the cardinality of an $r$-ball $B_r(u)$ centered at $u\in \mathbb{Z}_m^n$. It follows from the definition that $|B_r(u)|=1+\sum\limits_{i=1}\limits^{r}|S_i(u)|$. Let $I$ be an ideal in the pomset $\mathbb{P}$ with cardinality $i$ having exactly $j$ maximal elements $M(I)=\left\{C_{I}(a_1)/a_1,C_{I}(a_2)/a_2,\ldots,C_{I}(a_j)/a_j\right\}$. Then, the set
\begin{eqnarray*}
A_I&=&\{v=(v_1,v_2,\ldots,v_s)\in\mathbb{Z}_m^n: v_t=\bar{0}\ \text{if}\ t\notin I^{*};\\
&&v_{t_i}=\pm a, -a\leq v_{t_l}\leq a, 1\leq i\neq l\leq k_t\ \text{if}\ t\in M(I)^{*}, a=C_{I}(t);\\
&&v_{t_i}\in \mathbb{Z}_m, 1\leq i\leq k_t,\ \text{if}\ t\in I^{*}\setminus M(I)^{*}\}
\end{eqnarray*}
gives all vectors $v$ in $\mathbb{Z}_m^n$ such that $\langle supp_{(L,\pi)}(v)\rangle=I$. By the definition of $A_I$, we have $|A_I|=m^{\sum\limits_{t\in I^{*}\setminus M(I)^{*}}k_t}\prod\limits_{l=1}\limits^{j}N_{a_l}$ where $N_{a_l}$ is defined as follows:
\begin{itemize}
\item If $m$ is odd, $N_{a_l}=(2C_I(a_l)+1)^{k_{a_l}}-(2C_I(a_l)-1)^{k_{a_l}}$;
\item If $m$ is even, $N_{a_l}=\left\{
                             \begin{array}{ll}
                               \dfrac{(2C_I(a_l)+1)^{k_{a_l}}-(2C_I(a_l)-1)^{k_{a_l}}}{2}, & \text{if}\ C_{I}(a_l)=\left\lfloor\frac{m}{2}\right\rfloor; \\[3mm]
                               (2C_I(a_l)+1)^{k_{a_l}}-(2C_I(a_l)-1)^{k_{a_l}}, & \text{otherwise.}
                             \end{array}
                           \right.$
\end{itemize}
Note that for two distinct ideals $I_1$ and $I_2$, $A_{I_1}\cap A_{I_2}=\varnothing$. Denote by $\mathcal {I}_j^i$ the collection of all ideals in $\mathbb{P}$ with cardinality $i$ having exactly $j$ maximal elements, then $\bigcup\limits_{j=1}\limits^{\min\{i,n\}}\mathcal {I}_{j}^{i}=\mathcal {I}^{i}(\mathbb{P})$ and $\bigcup\limits_{j=1}\limits^{\min\{i,n\}}\bigcup\limits_{I\in\mathcal {I}_j^i}A_I$ gives all those vectors of $(\mathbb{P},\pi)$-weight $i$. Hence
$$|B_r(u)|=1+\sum\limits_{i=1}\limits^{r}\sum\limits_{j=1}\limits^{\min\{i,n\}}
\sum\limits_{I\in\mathcal {I}_j^i\neq\varnothing} |A_{I}|=1+\sum\limits_{i=1}\limits^{r}\sum\limits_{j=1}\limits^{\min\{i,n\}}
\sum\limits_{I\in\mathcal {I}_j^i\neq\varnothing}m^{\sum\limits_{t\in I^{*}\setminus M(I)^{*}}k_t}\prod\limits_{l=1}\limits^{j}N_{a_l}.$$

\subsection{$I$-balls}

\quad\; Let $(\mathbb{P},\pi)$ be a pomset block structure on $\mathbb{Z}_m^n$ and let $I$ be a submset of $M$. For $u\in \mathbb{Z}_m^n$, the \textbf{$I$-ball} centered at $u$ is defined as
$$B_I(u)=\left\{v\in\mathbb{Z}_m^n: supp_{(L,\pi)}(u-v)\subseteq I\right\}.$$
For $I\in \mathcal {I}(\mathbb{P})$, the $I$-ball centered at $u$ is equal to the set
$$B_I(u)=\left\{v\in\mathbb{Z}_m^n:\ \langle supp_{(L,\pi)}(u-v)\rangle\subseteq I\right\}.$$
Similarly, the \textbf{$I$-sphere} centered at $u$ is defined as
$$S_{I}(u)=\left\{v\in\mathbb{Z}_m^n:\ \langle supp_{(L,\pi)}(u-v)\rangle=I\right\}.$$
We denote the $I$-ball (resp. $I$-sphere) centered at $\bar{0}$ by $B_I$ (resp. $S_I$). Whenever required, we shall use the notation $B_{I,(\mathbb{P},\pi)}(u)$ (resp. $S_{I,(\mathbb{P},\pi)}(u)$) instead of $B_I(u)$ (resp. $S_I(u)$).

\begin{Definition}
Let $(\mathbb{P},\pi)$ be a pomset block structure on $\mathbb{Z}_m^n$ and $I$ be an ideal of $\mathbb{P}$. A $(\mathbb{P}, \pi)$-code $\mathcal {C}\subseteq \mathbb{Z}_m^n$ is said to be $I$-perfect if the $I$-balls centered at the codewords of $\mathcal {C}$ are pairwise disjoint and their union is $\mathbb{Z}_m^n$, that is,
$$\mathbb{Z}_m^n=\bigsqcup\limits_{u\in \mathcal {C}}B_I(u).$$
\end{Definition}

The following proposition is a generalization of [\ref{POMSET2}, Proposition 3] where the metric was considered as pomset metric. The proof is on similar lines and therefore we omit it.

\begin{proposition}\label{BI}
Let $(\mathbb{P},\pi)$ be a pomset block structure on $\mathbb{Z}_m^{n}$. If $I$ is an ideal with full count in $\mathbb{P}$, then
\begin{enumerate}[(1)]
  \item $B_I$ is a submodule of $\mathbb{Z}_m^n$ of dimension $\sum\limits_{i\in I^{*}}k_i$.
  \item For $u\in\mathbb{Z}_m^n$, $B_I(u)$ is a coset of $B_I$ containing $u$, that is, $B_I(u)=u+B_I$.
  \item For $u,v\in\mathbb{Z}_m^n$, $B_I(u)$ and $B_I(v)$ are either disjoint or identical. Moreover
  $$B_I(u)=B_I(v)\ \text{if and only if}\ supp_{(L,\pi)}(u-v)\subseteq I.$$
  \item $B_{I^{c},\widetilde{\mathbb{P}}}= B_{I,\mathbb{P}}^{\perp}$ where $\widetilde{\mathbb{P}}$ is the dual pomset of $\mathbb{P}$.
\end{enumerate}
\end{proposition}

\begin{remark}
\begin{enumerate}[(a)]
  \item The $(\mathbb{P},\pi)$-ball centered at $u\in\mathbb{Z}_m^n$ with radius $r$ is the union of $I$-balls centered at $u$ with $I\in\mathcal {I}^{r}(\mathbb{P})$, that is,
  $$B_r(u)=\bigcup\limits_{I\in\mathcal {I}^r(\mathbb{P})}B_I(u).$$
  \item Let $u$, $v$ be two vectors in $\mathbb{Z}_m^n$. If $u$ and $v$ belong to the same $I$-ball for some $I\in\mathcal {I}^r(\mathbb{P})$ with full count, then $d_{(\mathbb{P},\pi)}(u,v)\leq r$.
  \item Let $u$, $v$ be two vectors in $\mathbb{Z}_m^n$ and let $I$ be an ideal of $\mathbb{P}$ with partial count. If $B_I(u)\cap B_I(v)=\emptyset$, the inequality $d_{(\mathbb{P},\pi)}(u,v)> |I|$ is not necessarily true.
\end{enumerate}
\end{remark}

The following result can be easily obtained from the definition of $I$-perfect codes which is similar to the case of poset block metric (see [\ref{DASS}], Lemma 4.1).

\begin{lemma}\label{IPERFECT}
Let $(\mathbb{P},\pi)$ be a pomset block structure  on $\mathbb{Z}_m^n$. Let $\mathcal {C}\subseteq \mathbb{Z}_m^n$ be a linear $(n,m^k,d)$ $(\mathbb{P},\pi)$-code and $I$ be an ideal of $\mathbb{P}$ with full count. Then the following statements are equivalent:
\begin{enumerate}[(1)]
  \item $\mathcal {C}$ is an $I$-perfect code;
  \item $\sum\limits_{i\in I^{*}}k_i=n-k$ and $|\mathcal {C}\cap B_I|=1$;
  \item $|\mathcal {C}\cap B_I(u)|=1$ for all $u\in \mathbb{Z}_m^n$.
\end{enumerate}
\end{lemma}

Writing $u\in\mathbb{Z}_m^n$ as $(u_1,u_2)$ where $u_1\in\mathop{\oplus}\limits_{j\notin I^{*}}V_j$ and $u_2\in\mathop{\oplus}\limits_{i\in I^{*}}V_i$, we get the following.

\begin{proposition}
Let $(\mathbb{P},\pi)$ be a pomset block structure  on $\mathbb{Z}_m^n$ and $I$ be an ideal of $\mathbb{P}$ with full count. Then, an $(n,m^k,d)$ $(\mathbb{P},\pi)$-code $\mathcal {C}$ is $I$-perfect if and only if there exists a function
$$f:\mathop{\oplus}\limits_{j\notin I^{*}}V_j\rightarrow\mathop{\oplus}\limits_{i\in I^{*}}V_i$$
such that
$$\mathcal {C}=\left\{(v,f(v)): v\in\mathop{\oplus}\limits_{j\notin I^{*}}V_j\right\}.$$
\end{proposition}

\begin{proof}
Let $\mathcal {C}$ be an $I$-perfect $(n,m^k,d)$-code. Then for any $v\in\mathop{\oplus}\limits_{j\notin I^{*}}V_j$ there exists $c\in \mathcal {C}$ such that $(v,\bar{0})\in B_I(c)$ which implies that $c-(v,\bar{0})=(\bar{0},u)$ and hence $c=(v,u)$. Suppose that there exists another element $c'=(v,w)\in \mathcal {C}$. Then $c-c'=(\bar{0},u-w)\in B_I$ which implies that $c\in B_I(c')$, a contradiction to the fact that $\mathcal {C}$ is $I$-perfect. Thus, the function $f:\mathop{\oplus}\limits_{j\notin I^{*}}V_j\rightarrow\mathop{\oplus}\limits_{i\in I^{*}}V_i$ which sends $v\in \mathop{\oplus}\limits_{j\notin I^{*}}V_j$ to the unique $u\in\mathop{\oplus}\limits_{i\in I^{*}}V_i$ such that $c=(v,u)$ is well-defined. Moreover, $\left|\mathop{\oplus}\limits_{j\notin I^{*}}V_j\right|=m^k=|\mathcal {C}|$. We have $\mathcal {C}=\left\{(v,f(v)): v\in\mathop{\oplus}\limits_{j\notin I^{*}}V_j\right\}$.

On the other hand, if there exists such a function, then $\mathcal {C}\cap B_I(v,u)=\{(v,f(v)\}$ for any $(u,v)\in\mathbb{Z}_m^n$. We obtain that $\mathcal {C}$ is $I$-perfect by Lemma \ref{IPERFECT}.
\end{proof}

\begin{theorem}\label{Idual}
Let $(\mathbb{P},\pi)$ be a pomset block structure on $\mathbb{Z}_m^n$ and $I$ be an ideal with full count in $\mathbb{P}$. A linear $(n,m^k,d)$ $(\mathbb{P},\pi)$-code $\mathcal {C}$ is $I$-perfect if and only if $\mathcal {C}^{\bot}$ is an $I^{c}$-perfect $(\widetilde{\mathbb{P}},\pi)$-code where $\widetilde{\mathbb{P}}$ is the dual pomset of $\mathbb{P}$.
\end{theorem}

\begin{proof}
Let $\mathcal {C}$ be an $I$-perfect $(\mathbb{P},\pi)$-code. By Lemma \ref{IPERFECT}, $I^c$ is an ideal in $\widetilde{\mathbb{P}}$ with full count and satisfies that $\sum\limits_{i\in {I^{c}}^{*}}k_i=k$. Consider the $I^{c}$-ball centered at $\bar{0}\in \mathcal {C}^{\bot}$. If there exists another element $\bar{0}\neq c\in \mathcal {C}^{\bot}$ in $B_{I^c,(\widetilde{\mathbb{P}}, \pi)}$, then $c\in B_{I,(\mathbb{P},\pi)}^{\bot}$ by Proposition \ref{BI}. Let $v\in\mathbb{Z}_m^n$. Since $\mathcal {C}$ is an $I$-perfect code, there exists a unique $c'\in \mathcal {C}$ such that $v\in B_{I,(\mathbb{P},\pi)}(c')$ which implies that $v=c'+u$ for some $u\in B_{I,(\mathbb{P},\pi)}$. So, $c\cdot v=c\cdot c'+c\cdot u=0$. Since $v$ is arbitrary, we have $c=\bar{0}$. Therefore, $\left|\mathcal {C}^{\bot}\cap B_{I^{c},(\widetilde{\mathbb{P}},\pi)}\right|=1$ and hence $\mathcal {C}^{\bot}$ is $I^{c}$-perfect by Lemma \ref{IPERFECT}.
\end{proof}

\begin{theorem}\label{ERROR}
Let $\mathcal {C}$ be an $r$-error correcting $(\mathbb{P},\pi)$-code where $r\in\mathbb{N}$ is a multiple of $\left\lfloor\frac{m}{2}\right\rfloor$. Then for any $c,c'\in \mathcal {C}$, $c\neq c'$ and $I, I'\in\mathcal {I}^{r}(\mathbb{P})$ with full count, one has $c-c'\notin B_{I\oplus I'}$.
\end{theorem}

\begin{proof}
Assume that there exist $c,c'\in \mathcal {C}$, $c\neq c'$ and $I,I'\in\mathcal {I}^{r}(\mathbb{P})$ with full count satisfies $c-c'\in B_{I\oplus I'}$. Denote by  $p_{B_{I^{*}}}$ the projection of $\mathbb{Z}_m^n$ on blocks corresponding to $I^{*}$ and take $u=c-p_{B_{I^{*}}}(c-c')\in\mathbb{Z}_m^n$. Then
$$d_{(\mathbb{P},\pi)}(u,c)=w_{(\mathbb{P},\pi)}(p_{B_{I^{*}}}(c-c'))\leq \left\lfloor\frac{m}{2}\right\rfloor\cdot|I^{*}|=|I|=r.$$
Analogously, we have
$$d_{(\mathbb{P},\pi)}(u,c')=w_{(\mathbb{P},\pi)}\left( (c-c')-p_{B_{I^{*}}}(c-c') \right)
\leq w_{(\mathbb{P},\pi)}\left(p_{B_{{I'}^{*}}}(c-c')\right)\leq |I'^{*}|\cdot\left\lfloor\frac{m}{2}\right\rfloor=|I'|=r.$$
This yields that $u\in B_r(c)\cap B_r(c')$, a contradiction to the fact that $\mathcal {C}$ is $r$-error correcting.
\end{proof}

\begin{theorem}\label{error}
Let $\mathcal {C}$ be a $(\mathbb{P},\pi)$-code. If for any $c,c'\in \mathcal {C}$, $c\neq c'$ and $I, I'\in\mathcal {I}^{r}(\mathbb{P})$, one has $c-c'\notin B_{I\oplus I'}$. Then $\mathcal {C}$ is $r$-error correcting.
\end{theorem}

\begin{proof}
Assume that $\mathcal {C}$ is not an $r$-error correcting code. Then there exist $c,c'\in \mathcal {C}$, $c\neq c'$ and $u\in\mathbb{Z}_m^n$ such that $u\in B_r(c)\cap B_r(c')$. As $d_{(\mathbb{P},\pi)}(u,c)\leq r$ and $d_{(\mathbb{P},\pi)}(u,c')\leq r$, we have $|\langle supp_{(L,\pi)}(u-c)\rangle|\leq r$ and $|\langle supp_{(L,\pi)}(u-c')\rangle|\leq r$. Therefore, there exist $I,I'\in I^r(\mathbb{P})$ such that $\langle supp_{(L,\pi)}(u-c)\rangle\subseteq I$ and $\langle supp_{(L,\pi)}(u-c')\rangle\subseteq I'$ by Proposition \ref{POMSET1}. Hence $$supp_{(L,\pi)}(c-c')=supp_{(L,\pi)}(c-u+u-c')\subseteq supp_{(L,\pi)}(c-u)\oplus supp_{(L,\pi)}(u-c')\subseteq I\oplus I'.$$
This implies that $c-c'\in B_{I\oplus I'}$, a contradiction.
\end{proof}

The following can be easily obtained from the definitions of $r$-ball and $I$-ball.

\begin{proposition}\label{RPERFECT}
Let $\mathbb{P}$ be a pomset on $M=\left\{\left\lfloor\frac{m}{2}\right\rfloor/1,\left\lfloor\frac{m}{2}\right\rfloor/2, \ldots,\left\lfloor\frac{m}{2}\right\rfloor/s\right\}$ and $(\mathbb{P},\pi)$ a pomset block structure on $\mathbb{Z}_m^n$ such that $\mathcal {I}^r(\mathbb{P})=\{I\}$. Let $\mathcal {C}$ be a $(\mathbb{P},\pi)$-code. Then $\mathcal {C}$ is an $r$-perfect $(\mathbb{P},\pi)$-code if and only if $\mathcal {C}$ is an $I$-perfect $(\mathbb{P},\pi)$-code.
\end{proposition}

In the remainder of this section,  we consider an ideal of $\mathbb{P}$ with partial count. Before this, we give some basic facts on the partition of $\mathbb{Z}_m^t$ when $t$ is a positive integer.

Let $\epsilon\in\left[\lfloor\frac{m}{2}\rfloor-1\right]$. Let $E=[-\epsilon,\epsilon]$ and $E'=\{\epsilon+1,\epsilon+2,\ldots,m-\epsilon-1\}$. Set $$S=\left\{v=(v_1,\ldots,v_t)\in\mathbb{Z}_m^t: v_i\in E\ \text{for}\ 1\leq i\leq t\right\}.$$
Let $u=(u_1,\ldots,u_t)\in S\setminus \{\bar{0}\}$. Take $w=(w_1,\ldots,w_t)\in S$ satisfying
\begin{equation*}
  w_i=\left\{
                     \begin{array}{ll}
                     0,  & \text{if}\ u_i=0,\\[2mm]
                      \epsilon-(u_i-1), & \text{if}\ 1\leq u_i\leq \epsilon,\\[2mm]
                      -\epsilon-(u_i+1), & \text{if}\ -\epsilon\leq u_i\leq -1.
                     \end{array}
                   \right.
 \end{equation*}
Then there exists at least one $j\in[t]$ such that $u_j+w_j\in E'$. Thus $u+w\notin S$. On the other hand, we have $u+w\in u+S$ and $u\in u+S$. Hence we have the following result.

\begin{proposition}
Let $u\in S\setminus \{\bar{0}\}$. Then $S\cap (u+S)\neq \emptyset$ and $S\neq u+S$. \end{proposition}

Suppose that $m$ is divisible by $2\epsilon+1$. Let
$$T=\left\{v=(v_1,\ldots,v_t)\in\mathbb{Z}_m^t: v_i=t_i(2\varepsilon+1),\ 0\leq t_i\leq \frac{m}{2\epsilon+1}-1\ \text{for}\ 1\leq i\leq t\right\}.$$
It is known that $i(2\epsilon+1)+\beta\notin E$ for all $\beta\in E$ whenever $1\leq i\leq \frac{m}{2\varepsilon+1}-1$. It is also known that $i(2\epsilon+1)+E$, $j(2\epsilon+1)+E$ are disjoint for $0\leq i\neq j\leq \frac{m}{2\varepsilon+1}-1$ (see [\ref{POMSET2}]).

\begin{lemma}
Suppose that $m$ is divisible by $2\epsilon+1$. Then we have the followings:
\begin{enumerate}[(1)]
  \item For any $v\in S$ and $\bar{0}\neq u\in T$, we have $v+u\notin S$.
  \item Let $u\neq w\in T$, we have $u+S$ and $w+S$ are disjoint. Furthermore, $\bigsqcup\limits_{u\in T}(u+S)=\mathbb{Z}_m^t$.
\end{enumerate}
\end{lemma}

\begin{proof}
\begin{enumerate}[(1)]
\item Let $v=(v_1,\ldots,v_t)\in S$ and $\bar{0}\neq u=(u_1,\ldots,u_t)\in T$. Then there exists $i\in [t]$ such that $u_i=t_i(2\epsilon+1)\neq 0$. Note that $1\leq t_i\leq \frac{m}{2\varepsilon+1}-1$. It follows from $v_i+u_i=v_i+t_i(2\epsilon+1)\notin E$ that $u+v\notin S$.
\item Let $u=(u_1,\ldots,u_t),\ w=(w_1,\ldots,w_t)\in T$. Suppose that $(u+S)\cap (w+S)\neq\emptyset$. Then there exist $v=(v_1,\ldots,v_t)\in S$ and $v'=(v'_1,\ldots,v'_t)\in S$ such that $u+v=w+v'$. Since $u\neq w$, there exists $i\in [t]$ such that $u_i\neq w_i$. Thus $u_i+v_i=w_i+v'_i$, which contradicts to the fact that $u_i+E$ and $w_i+E$ are disjoint. The result then follows.
\end{enumerate}
\end{proof}

\begin{corollary}
If $m$ is not divisible by $2\epsilon+1$, then the translates of $S$ can not form a partition of $\mathbb{Z}_m^t$.
\end{corollary}

We now let $I\in\mathcal {I}(\mathbb{P})$ be an ideal with partial count. Suppose that $I_p=\{l_1,l_2,\ldots,l_{\lambda}\}\subseteq I^{*}$ is the collection of the elements in $I^{*}$ which has partial count in $I$ and suppose that $I_f=\{r_1,r_2,\ldots,r_{\delta}\}$ is the collection of the elements in $I^{*}$ which has full count in $I$. Then $\lambda+\delta=\left|I^{*}\right|$. Since the pomset block metric is translation invariant, that is, for all $u,v,w\in\mathbb{Z}_m^n$, $d_{(\mathbb{P},\pi)}(u,v)=d_{(\mathbb{P},\pi)}(u+w,v+w)$, we have that $B_I(u)=u+B_I$. Note that
\begin{eqnarray}\label{partial}
\nonumber B_I&=&\left\{v=(v_1,v_2,\ldots,v_s)\in\mathbb{Z}_m^n: v_t=\bar{0}\ \text{if}\ t\notin I^{*};\right.\\
&& -a \leq v_{t_j}\leq a, a=C_{I}(t), 1\leq j\leq k_{t}\ \text{if}\ t\in I_p;\\
\nonumber &&\left. v_{t_j}\in \mathbb{Z}_m, 1\leq j\leq k_t,\ \text{if}\ t\in I_f\right\}
\end{eqnarray}
and hence
$$\left| B_I\right|=\left(1+2C_I(l_1)\right)^{k_{l_1}}\left(1+2C_I(l_2)\right)^{k_{l_2}}\cdots \left(1+2C_I(l_{\lambda})\right)^{k_{l_{\lambda}}}m^{\sum\limits_{j\in I_f}k_j}.$$

\begin{remark}
It follows from (\ref{partial}) that $B_I$ is not a subgroup of $\mathbb{Z}_m^n$.
\end{remark}

With the notations given above, we have the following result which can be easily obtained by counting argument.
\begin{theorem}
If $m$ is divisible by $2C_I(l_i)+1$ for all $l_i\in I_p$ then the $I$-balls centered at the elements in
\begin{eqnarray*}
D&=&\left\{v=(v_1,v_2,\ldots,v_s)\in\mathbb{Z}_m^n: v_t=\bar{0}\ \text{if}\ t\in I_f;\right.\\[2mm]
&&v_{t_i}=j(2C_I(t)+1), 0\leq j\leq \frac{m}{2C_I(t)+1}-1, 1\leq i\leq k_t\ \text{if}\ t\in I_p;\\[2mm]
&&\left.v_{t_i}\in \mathbb{Z}_m, 1\leq i\leq k_t,\ \text{if}\ t\in [s]\setminus I^{*}\right\}
\end{eqnarray*}
partition the space $\mathbb{Z}_m^n$. Moreover, we have $|D|=\mathop{\Pi}\limits_{t\in I_p}\left(\frac{m}{2C_I(t)+1}\right)^{k_t}m^{\sum\limits_{i\in[s]\setminus I^{*}}k_i}$.
\end{theorem}

\begin{corollary}
If $m$ is not divisible by $2C_I(j)+1$ for some $j\in I_p$, then no collection of $I$-balls will partition $\mathbb{Z}_m^n$.
\end{corollary}

\begin{remark}\label{pball}
Let $\mathcal {C}\subseteq\mathbb{Z}_m^n$ be an $(n,K,d)$ $(\mathbb{P},\pi)$-code.
\begin{enumerate}[(1)]
\item
If $\mathcal {C}$ is an $r$-perfect code, then for any $I\in \mathcal {I}^r(\mathbb{P})$, the $I$-balls centered at the codewords of $\mathcal {C}$ are disjoint.
\item
If a $(\mathbb{P},\pi)$-code $\mathcal {C}$ is $I$-perfect for some $I\in\mathcal {I}^r(\mathbb{P})$ with partial count, we can not guarantee that $\mathcal {C}$ is $r$-perfect since there may exist an ideal $I'\in\mathcal {I}^{r}(\mathbb{P})$ with partial count and $u,v\in\mathcal {C}$ such that $B_{I'}(u)\cap B_{I'}(v)\neq\emptyset$ (Example \ref{pr} will illustrate a counterexample).
\item
Let $m$ be a  prime and let $I$ be an ideal of $\mathbb{P}$ with partial count, then $\mathcal {C}$ can not be $I$-perfect $(\mathbb{P},\pi)$ for any $I\in\mathcal {I}(\mathbb{P})$ with partial count.
\end{enumerate}
\end{remark}

\begin{example}\label{pr}
Let $\mathbb{P}=(M,R)$ be a pomset where $M=\{3/1,3/2\}$ and $$R=\{(9/(3/1,3/1),9/(3/2,3/2)\}.$$
Let $\pi$ be a labeling of the pomset $\mathbb{P}$ such that $\pi(1)=2$ and $\pi(2)=1$. For $I=\{1/1,3/2\}$, by the above discussion, we can find an $I$-perfect $(\mathbb{P},\pi)$-code $\mathcal {C}=\{(0,0,0),(3,0,0),(0,3,0),(3,3,0)\}\subseteq\mathbb{Z}_6^2$. Consider $I'=\{2/1,2/2\}\in\mathcal {I}^4(\mathbb{P})$. We observe that $(2,1,0)\in B_{I'}((3,0,0))\cup B_{I'}((0,3,0))$. Since $B_r(u)=\bigcup\limits_{I\in\mathcal {I}^r(\mathbb{P})}B_I(u)$, $\mathcal {C}$ is not 4-perfect. If $\mathbb{P}$ is a chain pomset with order relation $3/2\ R\ 3/1$, then $\mathcal {I}^4(\mathbb{P})=\{I\}$ and $\mathcal {C}$ is a 4-perfect $(\mathbb{P},\pi)$-code now.
\end{example}

Here is an example for a code being $r$-perfect but not $I$-perfect for any $I$ with partial count.

\begin{example}
Let $\mathbb{P}=(M,R)$ be a pomset where $M=\{2/1,2/2\}$ and $$R=\{(4/(2/1,2/1),4/(2/2,2/2)\}.$$
Let $\pi$ be a labeling of the pomset $\mathbb{P}$ such that $\pi(1)=\pi(2)=1$. Let $r=1$. We have $\mathcal {I}^1(\mathbb{P})=\{I_1,I_2\}$ where $I_1=\{1/1\}$, $I_2=\{1/2\}$. Then $B_1(\bar{0})=\{(0,0),(0,1),(1,0),(0,4),(4,0)\}$. Consider the $(\mathbb{P},\pi)$-code $\mathcal {C}=\{(0,0),(1,2),(2,4),(3,1),(4,3)\}\subseteq\mathbb{Z}_5^2$. It is routine to verify that $\mathcal {C}$ is 1-perfect code. On the other hand, $\mathcal {C}$ can not be $I$-perfect for any $I\in\mathcal {I}(\mathbb{P})$ with partial count by the above discussion.
\end{example}

\subsection{MDS $(\mathbb{P},\pi)$-code}

\begin{theorem}\label{linear}(Singleton Bound)
Let $(\mathbb{P},\pi)$ be a pomset block structure on $\mathbb{Z}_m^n$ and $\mathcal {C}\subseteq \mathbb{Z}_m^n$ be an $(n,K,d)$ $(\mathbb{P},\pi)$-code. Denote by $r=\left\lfloor\frac{d-1}{\left\lfloor\frac{m}{2}\right\rfloor}\right\rfloor$. Then
\begin{equation}\label{BOUND}
  n-\lceil \log_mK\rceil\geq \max\limits_{I\in\mathcal{I}(\mathbb{P}),|I^{*}|=r}\sum\limits_{i\in I^*}k_i.
\end{equation}
\end{theorem}

\begin{proof}
Let $I\in\mathcal {I}(\mathbb{P})$ be an ideal of $\mathbb{P}$ with $|I^{*}|=r$. We may assume that $I$ is full count, otherwise one can increase the counts of the maximal elements with partial count in $I$ to $\left\lfloor\frac{m}{2}\right\rfloor$. Take $u,v\in \mathcal {C}$. If $u_i=v_i$ for all $i\in[s]\setminus I^*$. Then $d_{(\mathbb{P},\pi)}(u,v)\leq |I|\leq d-1$, a contradiction.
This means that any two distinct codewords of $\mathcal {C}$ will differ in at least one position outside $I^*$. Therefore there exists an injective map from $\mathcal {C}$ to $\mathbb{Z}_m^{n-\sum\limits_{i\in I^{*}}k_i}$ which implies that $\log_mK\leq n-\sum\limits_{i\in I^{*}}k_i$. Hence $\lceil\log_mK\rceil\leq n-\sum\limits_{i\in I^{*}}k_i$. From this, we get inequality (\ref{BOUND}).
\end{proof}

\begin{remark}
Note that when $k_i=1$ for all $1\leq i\leq s$, inequality (\ref{BOUND}) would be
$$n-\lceil\log_mK\rceil\geq \left\lfloor\frac{d-1}{\left\lfloor\frac{m}{2}\right\rfloor}\right\rfloor.$$
This is the Singleton bound for pomset code, see [\ref{POMSET2}, Theorem 2].
\end{remark}

\begin{Definition}
Let $(\mathbb{P},\pi)$ be a pomset block structure on $\mathbb{Z}_m^n$. A $(\mathbb{P},\pi)$-code $\mathcal {C}$ is said to be a maximum distance separable \textup{(MDS)} $(\mathbb{P},\pi)$-code if it attains the Singleton bound.
\end{Definition}

\begin{example}
Let $\mathbb{P}=(M,R)$ be a pomset on $M=\{2/1,2/2,2/3,2/4\}$ and
$$R=\{4/(2/1,2/1),4/(2/2,2/2),4/(2/3,2/3),4/(2/4,2/4),4/(2/1,2/2),4/(2/3,2/4)\}.$$
Let $\pi$ be a labeling of the pomset $\mathbb{P}$ such that $\pi(1)=\pi(3)=1$ and $\pi(2)=\pi(4)=2$. Consider the $(\mathbb{P},\pi)$-code $\mathcal {C}\subseteq\mathbb{Z}_5^6$ generated by the following matrix:
$$\left(
  \begin{array}{cccccc}
    1 & 0 & 2 & 2 & 0 & 1 \\
    0 & 2 & 4 & 1 & 1 & 0 \\
  \end{array}
\right).$$
Then the code $\mathcal {C}$ is a linear $(\mathbb{P},\pi)$-code of length 6 with $d_{(\mathbb{P},\pi)}(\mathcal {C})=7$. We have $\left\lfloor\frac{d-1}{\left\lfloor\frac{m}{2}\right\rfloor}\right\rfloor=3$. The ideals in $\mathbb{P}$ such that $|I^*|=3$ are
$$\Big\{\{2/1,1/2,1/3\},\{2/1,1/2,2/3\},\{2/1,2/2,1/3\},\{2/1,2/2,2/3\},$$
$$\{1/1,2/3,1/4\}, \{1/1,2/3,2/4\},\{2/1,2/3,1/4\},\{2/1,2/3,2/4\}\Big\}.$$
Thus
$$\max_{I\in\mathcal {I}(\mathbb{P}),|I^*|=3}\sum_{i\in I^*}k_i=4=n-k.$$
Therefore $\mathcal {C}$ is an MDS code.
\end{example}

\begin{remark}
Let $\mathcal {C}$ be an MDS $(n,K,d)$ $(\mathbb{P},\pi)$-code. Note that there always exists an ideal $I\in\mathcal {I}(\mathbb{P})$ with full count whose cardinality is $\left\lfloor\frac{m}{2}\right\rfloor\cdot \left\lfloor\frac{d-1}{\left\lfloor\frac{m}{2}\right\rfloor}\right\rfloor$ such that $\sum\limits_{i\in I^{*}}k_i=n-\lceil\log_mK\rceil$. Oterwise, assume that $I\in\mathcal {I}(\mathbb{P})$ is an ideal with partial count such that $\sum\limits_{i\in I^{*}}k_i=n-\lceil\log_mK\rceil$. By increasing the counts of the maximal elements with partial counts in $I$, one can get the ideal $J$ with full count whose cardinality is $\left\lfloor\frac{m}{2}\right\rfloor\cdot \left\lfloor\frac{d-1}{\left\lfloor\frac{m}{2}\right\rfloor}\right\rfloor$ and satisfies that $\sum\limits_{i\in I^{*}}k_i=n-\lceil\log_mK\rceil$.
\end{remark}

Let $m$ be a prime, that is, $\mathbb{Z}_m$ is a field. Let $\mathcal {C}$ be a linear $(n,m^k,d)$ $(\mathbb{P},\pi)$-code. A generator matrix $G$ and a parity check matrix $H$ of $\mathcal {C}$ are defined as in the classical case. The parity check matrix $H$ can be viewed as $H=[H_1\ H_2\ \cdots\ H_s]$ where $H_i$ is an $(n-k)\times k_i$ matrix. The set of blocks $H_{l_1}, H_{l_2},\ldots, H_{l_r}$ is called linearly independent if, for $\alpha_i\in V_{l_i}$,
$$H_{l_1}\alpha_1+H_{l_2}\alpha_2+\cdots+H_{l_r}\alpha_r=\bar{0}$$
deduces $\alpha_i=\bar{0}$ for all $i\in[r]$.
Otherwise the block set $H_{l_1}, H_{l_2}, \ldots, H_{l_r}$ is called linearly dependent.

Let $\mathbb{P}=(M,R)$ be a pomset defined on the multiset $M=\left\{\left\lfloor\frac{m}{2}\right\rfloor/1,\left\lfloor\frac{m}{2}\right\rfloor/2, \ldots,\left\lfloor\frac{m}{2}\right\rfloor/s\right\}$. We can define a corresponding poset $P$ with the underlying set $\{1,\ldots,s\}$ whose order relation is given by
\begin{center}
$a\leq b$ in $P$ if and only if $p/a\ R\ q/b$ in $\mathbb{P}$.
\end{center}
Given a subset $Q\subseteq [s]$, we denote by $<Q>_P$ the smallest ideal of $P$ containing $Q$. With these notations, we have the following.

\begin{theorem}
    Let $m$ be a prime. Let $\mathcal {C}$ be a linear $(n,m^k,d)$ $(\mathbb{P},\pi)$-code and let $H$ be a parity check matrix of $\mathcal {C}$. Then $\mathcal {C}$ has a codeword $c$ such that $\left|\langle supp_{(L,\pi)}(c)\rangle^{*}\right|=t$ if and only if there exists an ideal $I$ of $\mathbb{P}$ with $|I^{*}|=t$ satisfying that the blocks of $H$ corresponding to $I^{*}$ are linearly dependent.
\end{theorem}

\begin{proof}
Let $c=(c_1,c_2,\cdots,c_s)\in \mathcal {C}$ satisfying $\left|\langle supp_{(L,\pi)}(c)\rangle^{*}\right|=t$. Let $I=\langle supp_{(L,\pi)}(c)\rangle$. Then $c_i\neq \bar{0}$ for $i\in M(I)^{*}$ and $c_i=\bar{0}$ for $i\notin I^{*}$. It follows from
$$H_1c_1+H_2c_2+\cdots+H_sc_s=\bar{0}$$
that the blocks of $H$ corresponding to $I^{*}$ are linearly dependent.

On the other hand, we suppose that the set of blocks $H_{l_1},H_{l_2},\ldots,H_{l_{\lambda}}$ are linearly dependent and suppose that there exist $\bar{0}\neq\alpha_{l_i}\in V_{l_i}$ such that
$$H_{l_1}\alpha_{l_1}+H_{l_2}\alpha_{l_2}+\cdots+H_{l_j}\alpha_{l_{\lambda}}=\bar{0}.$$
Let $I$ be an ideal of $\mathbb{P}$ whose root set is $I^{*}=<\{l_1,\ldots,l_{\lambda}\}>_P$.
Take $c=(c_1,c_2,\ldots,c_s)\in \mathbb{Z}_m^n$ such that
 \begin{equation*}
  c_{l_i}=\left\{
                     \begin{array}{ll}
                      \alpha_{l_i}, & \text{if}\ i\in[\lambda],\\[2mm]
                      \bar{0}, & \text{otherwise}.
                     \end{array}
                   \right.
 \end{equation*}
Then $Hc=\bar{0}$ which implies that $c\in \mathcal {C}$. Furthermore, $\langle supp_{(L,\pi)}(c)\rangle^{*}=I^*$. This completes the proof.
\end{proof}

\begin{corollary}\label{field}
Let $m$ be a prime. Let $\mathcal {C}$ be a linear $(\mathbb{P},\pi)$-code and let $H$ be a parity check matrix of $\mathcal {C}$.
Then
$$\min\left\{\left|\langle supp_{(L,\pi)}(c)\rangle^{*}\right|,\ c\in \mathcal {C}\right\}=t$$
 if and only if
 $$t=\min\left\{j: I\in \mathcal {I}(\mathbb{P})\ \text{satisfies}\ |I^{*}|=\{l_1,\ldots,l_j\},\ rank\{H_{l_1},\ldots,H_{l_j}\}<\sum\limits_{i=1}^j k_{l_i}\right\}.$$
\end{corollary}

Since
$$\left\lfloor\frac{d}{\left\lfloor\frac{m}{2}\right\rfloor} \right\rfloor\leq\min\left\{\left|\langle supp_{(L,\pi)}(c)\rangle^{*}\right|,\ c\in \mathcal {C}\right\},$$
we have the following.
\begin{remark}
Let $m$ be a prime. Let $\mathcal {C}$ be a linear $(\mathbb{P},\pi)$-code and let $H$ be a parity check matrix of $\mathcal {C}$. Then
$$\left\lfloor\frac{d}{\left\lfloor\frac{m}{2}\right\rfloor} \right\rfloor\leq\min\left\{j: I\in \mathcal {I}(\mathbb{P})\ \text{satisfies}\ |I^{*}|=\{l_1,\ldots,l_j\},\ rank\{H_{l_1},\ldots,H_{l_j}\}<\sum\limits_{i=1}^j k_{l_i}\right\}$$
by Corollary \ref{field}. The equality holds if and only if $d$ is divisible by $\left\lfloor\frac{m}{2}\right\rfloor$.
\end{remark}

\begin{theorem}\label{MDSSN}
Let $(\mathbb{P},\pi)$ be a pomset block structure on $\mathbb{Z}_m^n$ and $\mathcal {C}\subseteq\mathbb{Z}_m^n$ be an $(n,m^k,d)$ $(\mathbb{P},\pi)$-code. Then $\mathcal {C}$ is an MDS $(\mathbb{P},\pi)$-code if and only if $\mathcal {C}$ is $I$-perfect for some ideals $I\in\mathcal {I}^{\left\lfloor\frac{m}{2}\right\rfloor\cdot \left\lfloor\frac{d-1}{\left\lfloor\frac{m}{2}\right\rfloor}\right\rfloor}$ with full count.
\end{theorem}

\begin{proof}
Suppose that $\mathcal {C}$ is an MDS $(\mathbb{P},\pi)$-code. Denote by $r=\left\lfloor\frac{d-1}{\left\lfloor\frac{m}{2}\right\rfloor}\right\rfloor$. Then
$$n-k=\max\limits_{I\in\mathcal{I}(\mathbb{P}),|I^{*}|=r}\sum\limits_{i\in I^*}k_i.$$
Let $I$ be an ideal of $\mathbb{P}$  with full count whose cardinality is $r\cdot\left\lfloor\frac{m}{2}\right\rfloor$ such that $n-k=\sum\limits_{i\in I^{*}}k_i$. Assume there exist two codewords $c,c'\in \mathcal {C} $ such that the $I$-balls centered at $c$ and $c'$ have nonempty intersection. It follows from Proposition \ref{BI} that $B_I(c)=B_I(c')$. Then $c\in B_I(c')$ which implies that $\langle supp_{(L,\pi)}(c-c')\rangle\subseteq I$. Therefore
$$d_{(\mathbb{P},\pi)}(c,c')\leq |I|=\left\lfloor\frac{m}{2}\right\rfloor\cdot \left\lfloor\frac{d-1}{\left\lfloor\frac{m}{2}\right\rfloor}\right\rfloor\leq \left\lfloor\frac{m}{2}\right\rfloor\cdot\frac{d-1}{\left\lfloor\frac{m}{2}\right\rfloor} =d-1,$$
a contradiction. Thus, any two $I$-balls centered at distinct codewords of $\mathcal {C}$ must be disjoint. Hence, $\bigsqcup\limits_{c\in \mathcal {C}}B_I(c)$ is a disjoint union and it contains $|\mathcal {C}|\cdot|B_I|=m^n$ elements which implies that $\mathcal {C}$ is $I$-perfect.

Conversely, let $\mathcal {C}$ be an $I$-perfect code for some $I\in\mathcal {I}^{r\cdot\lfloor\frac{m}{2}\rfloor}$ with full count. It follows from Lemma \ref{IPERFECT} that $\sum\limits_{i\in I^{*}}k_i=n-k$ and hence
$$\max\limits_{I\in\mathcal{I}(\mathbb{P}),|I^{*}|=r}\sum\limits_{i\in I^*}k_i=n-k$$
This completes the proof.
\end{proof}

\begin{remark}
An $(n,m^k,d)$ MDS $(\mathbb{P},\pi)$-code $\mathcal {C}$ is $I$-perfect for all $I\in\mathcal {I}^{r}(\mathbb{P})$ with full count such that $\sum\limits_{i\in I^{*}}k_i=n-k$ where $r=\left\lfloor\frac{m}{2}\right\rfloor\cdot \left\lfloor\frac{d-1}{\left\lfloor\frac{m}{2}\right\rfloor}\right\rfloor$.
\end{remark}

Since an ideal of $\mathbb{P}$ with partial count is always contained in an ideal of $\mathbb{P}$ with full count, we have the following.

\begin{corollary}
Let $\mathcal {C}\subseteq\mathbb{Z}_m^n$ be an $(n,K,d)$ $(\mathbb{P},\pi)$-code. Let $I$ be an ideal of $\mathbb{P}$ such that $|I^{*}|=\left\lfloor\frac{d-1}{\left\lfloor\frac{m}{2}\right\rfloor}\right\rfloor$. Then $I$-balls centered at the codewords of $\mathcal {C}$ are disjoint.
\end{corollary}

Here we give an example to show the existence of a $(\mathbb{P},\pi)$-code which is MDS but not $I$-perfect for any $I\in\mathcal {I}(\mathbb{P})$ with partial count.

\begin{example}
Let $\mathbb{P}=(M,R)$ be a pomset where $M=\{2/1,2/2\}$ and $$R=\{(4/(2/1,2/1),4/(2/2,2/2),4/(2/2,2/1)\}.$$
Let $\pi$ be a labeling of the pomset $\mathbb{P}$ such that $\pi(1)=2$ and $\pi(2)=1$. Consider the $(\mathbb{P},\pi)$-code $\mathcal {C}\subseteq\mathbb{Z}_5^3$ generated by the following matrix:
$$\left(
  \begin{array}{ccc}
    0 & 1 & 3 \\
    1 & 2 & 0 \\
  \end{array}
\right).$$
Then the code $\mathcal {C}$ is a linear $(\mathbb{P},\pi)$-code over $\mathbb{Z}_5$ of length 3 with $d_{(\mathbb{P},\pi)}(\mathcal {C})=3$. We have $\left\lfloor\frac{d_{(\mathbb{P},\pi)}(\mathcal {C})-1}{\left\lfloor\frac{m}{2}\right\rfloor}\right\rfloor=1$. The ideals of $\mathbb{P}$ such that $|I^*|=1$ are $I_1=\{1/2\}$ and $I_2=\{2/2\}$ which implies that $\mathcal {C}$ is an MDS $(\mathbb{P},\pi)$-code. The partial count ideal in $\mathbb{P}$ is $I_1=\{(1/2)\}$ and $I_2=\{1/1,2/2\}$. It is routine to verify that $I_1$-balls centered at the codewords of $\mathcal {C}$ are disjoint and the union of $I_1$-balls contains 75 elements. Take $u=(0,1,3)$ and $v=(1,2,0)\in \mathcal {C}$. It is easy to see that $B_{I_2}(u)\cap B_{I_2}(v)\neq \phi$. Therefore, for any ideal $I$  of $\mathbb{P}$ with partial count, the code $\mathcal {C}$ is not $I$-perfect.
\end{example}

If $I$ is an ideal of $\mathbb{P}$ with partial count, then an $I$-perfect code is not necessarily an MDS code, as we can see in the next example.

\begin{example}\label{D}
Let $\mathbb{P}=(M,R)$ be a pomset where $M=\{3/1,3/2\}$ and $$R=\{(9/(3/1,3/1),9/(3/2,3/2)\}.$$
Let $\pi$ be a labeling of the pomset $\mathbb{P}$ such that $\pi(1)=1$ and $\pi(2)=2$. Consider the $(\mathbb{P},\pi)$-code $\mathcal {C}\subseteq\mathbb{Z}_6^3$ defined by $\mathcal {C}=\{(0,0,0),(0,3,0),(0,0,3),(0,3,3)\}$. Consider an ideal $I=\{3/1,1/2\}$ of $\mathbb{P}$ which has partial count. It is routine to verify that $\mathcal {C}$ is $I$-perfect. On the other hand, we have that $d_{(\mathbb{P},\pi)}(\mathcal {C})=3$. Then $\left\lfloor\frac{d_{(\mathbb{P},\pi)}(\mathcal {C})-1}{\left\lfloor\frac{m}{2}\right\rfloor}\right\rfloor= 0$. Therefore $\mathcal {C}$ is not MDS.

\end{example}

\begin{example}\label{d}
Let $\mathbb{P}=(M,R)$ be a pomset where $M=\{3/1,3/2\}$ and $$R=\{(9/(3/1,3/1),9/(3/2,3/2),9/(3/1,3/2)\}.$$
Let $\pi$ be a labeling of the pomset $\mathbb{P}$ such that $\pi(1)=\pi(2)=1$. Consider the $(\mathbb{P},\pi)$-code $\mathcal {C}\subseteq\mathbb{Z}_6^2$ defined by $\mathcal {C}=\{(0,0),(1,3)\}$. Consider $I=\{3/1,1/2\}$ of $\mathbb{P}$ which has partial count. It is routine to verify that $\mathcal {C}$ is $I$-perfect. Note that $d_{(\mathbb{P},\pi)}(\mathcal {C})=6$. Then $\left\lfloor\frac{d_{(\mathbb{P},\pi)}(\mathcal {C})-1}{\left\lfloor\frac{m}{2}\right\rfloor}\right\rfloor=1$. The ideals of $\mathbb{P}$ such that $|I^*|=1$ are $I_1=\{1/1\}$ and $I_2=\{2/1\}$.  Therefore $\mathcal {C}$ is an MDS $(\mathbb{P},\pi)$-code.
\end{example}

By comparing Examples \ref{D} and  \ref{d}, we speculate that being a chain pomset may be a necessary condition for an $I$-perfect code to be an MDS code. We will prove this conjecture later (Theorem \ref{CHAIN}).

\section{Duality and weight distribution}

\quad\; In this section, we consider the case for an $(n,K,d)$ MDS $(\mathbb{P},\pi)$-code when all the blocks have the same dimension. Let $(\mathbb{P},\pi)$ be a pomset block structure on $\mathbb{Z}_m^n$ such that $k_1=k_2=\cdots=k_s=t$. Then the Singleton bound becomes
\begin{equation}\label{bound}
n-\lceil\log_mK\rceil\geq t\cdot\left\lfloor\frac{d-1}{\left\lfloor\frac{m}{2}\right\rfloor}\right\rfloor.
\end{equation}

\subsection{Duality}

\quad\; In what follows, we characterize $(n,K,d)$ MDS $(\mathbb{P},\pi)$-codes when all the blocks have the same dimension.

\begin{theorem}\label{MDSdual}
Let $\mathbb{P}$ be a pomset on $M=\left\{\left\lfloor\frac{m}{2}\right\rfloor/1,\left\lfloor\frac{m}{2}\right\rfloor/2, \ldots,\left\lfloor\frac{m}{2}\right\rfloor/s\right\}$ and $\widetilde{\mathbb{P}}$ be its dual pomset on $M$. Let $\pi$ be a labeling of the pomset $\mathbb{P}$ with $k_1=k_2=\cdots=k_s=t$ and $\mathcal {C}$ be a linear $(n,m^k,d)$ $(\mathbb{P},\pi)$-code. The the following statements are equivalent:
\begin{enumerate}[(1)]
  \item $\mathcal {C}$ is an MDS $(\mathbb{P},\pi)$-code;
  \item $\mathcal {C}$ is an $I$-perfect $(\mathbb{P},\pi)$-code for all $I\in\mathcal {I}^{\frac{n-k}{t}\cdot\left\lfloor\frac{m}{2}\right\rfloor}(\mathbb{P})$ with full count;
  \item $\mathcal {C}^{\bot}$ is an $I$-perfect $(\widetilde{\mathbb{P}},\pi)$-code for all $I\in\mathcal {I}^{\frac{k}{t}\cdot\left\lfloor\frac{m}{2}\right\rfloor}(\widetilde{\mathbb{P}})$ with full count;
  \item $\mathcal {C}^{\bot}$ is an MDS $(\widetilde{\mathbb{P}}, \pi)$-code;
  \item For any ideal $I$ with full count of $\mathbb{P}$ and $x\in\mathbb{Z}_m^n$,
  \begin{equation}\label{duality}
  |B_I(x)\cap \mathcal {C}|=\left\{
                     \begin{array}{ll}
                       m^{tl-n+k}, & \text{if}\ |I|=l\cdot\lfloor\frac{m}{2}\rfloor \geq \frac{n-k}{t}\cdot\left\lfloor\frac{m}{2}\right\rfloor,\\[2mm]
                       1, & \text{if}\ |I|=l\cdot\left\lfloor\frac{m}{2}\right\rfloor< \frac{n-k}{t}\cdot\left\lfloor\frac{m}{2}\right\rfloor\ \text{and}\ x\in\bigcup\limits_{c\in \mathcal {C}}B_I(c),\\[2mm]
                       0, & \text{if}\ |I|=l\cdot\left\lfloor\frac{m}{2}\right\rfloor< \frac{n-k}{t}\cdot\left\lfloor\frac{m}{2}\right\rfloor\ \text{and}\ x\notin\bigcup\limits_{c\in \mathcal {C}}B_I(c).
                     \end{array}
                   \right.
  \end{equation}
\end{enumerate}
\end{theorem}

\begin{proof}
$(1)\Leftrightarrow (2)$ follows from Theorem \ref{MDSSN}, $(2)\Leftrightarrow (3)$ follows from Theorem \ref{Idual} and $(3)\Leftrightarrow (4)$ follows from Theorem \ref{MDSSN}. We only need to show that (1) is equivalent to (5).

Assume that $\mathcal {C}$ is an MDS $(\mathbb{P},\pi)$-code. Let $I$ be an ideal of $\mathbb{P}$ with full count such that $|I|=l\cdot\left\lfloor\frac{m}{2}\right\rfloor$ and $x\in\mathbb{Z}_m^n$. There are two cases.
\begin{itemize}
  \item {\bf Case 1:} Suppose that $l\cdot\left\lfloor\frac{m}{2}\right\rfloor<\frac{n-k}{t}\cdot \left\lfloor\frac{m}{2}\right\rfloor$. Note that no two codewords of $\mathcal {C}$ belong to the same $I$-ball. Otherwise, there exist $c,c'\in \mathcal {C}$, $c\neq c'$ such that $c\in B_I(c')$. This implies that $\langle supp_{(L,\pi)}(c-c')\rangle\subseteq I$ and thus
      $$d_{(\mathbb{P},\pi)}(c,c')=\left|\langle supp_{(L,\pi)}(c-c')\rangle\right|\leq \left|I\right|=l\cdot\left\lfloor\frac{m}{2}\right\rfloor< \frac{n-k}{t}\cdot\left\lfloor\frac{m}{2}\right\rfloor= \left\lfloor\frac{d-1}{\left\lfloor\frac{m}{2}\right\rfloor}\right\rfloor\cdot \left\lfloor\frac{m}{2}\right\rfloor\leq d-1,$$
  a contradiction. By Proposition \ref{BI}, any two $I$-balls are either disjoint or identical and hence
    $$\left|B_I(x)\cap \mathcal {C}\right|=\left\{
                          \begin{array}{ll}
                            1, & \text{if}\ x\in\bigcup\limits_{c\in \mathcal {C}}B_I(c),\\[2mm]
                            0,  & \text{if}\ x\notin\bigcup\limits_{c\in \mathcal {C}}B_I(c).
                          \end{array}
                        \right.$$
  \item {\bf Case 2:} Suppose that $l\cdot\left\lfloor\frac{m}{2}\right\rfloor\geq \frac{n-k}{t}\cdot\left\lfloor\frac{m}{2}\right\rfloor$. Since $I$ is full count with $|I^{*}|=l$, there exists an ideal $J$ of $\mathbb{P}$ with full count such that $|J^{*}|=\frac{n-k}{t}$ and $J\subseteq I$. It follows from Theorem \ref{MDSSN} that $\mathcal {C}$ is $J$-perfect. Note that $B_J$ is a submodule of $B_I$. The number of cosets of $B_J$ in $B_I$ is $m^{lt-n+k}$. By Lemma \ref{IPERFECT}, every $B_J(u)$ contains exactly one element of $\mathcal {C}$. Therefore, $|B_I(x)\cap \mathcal {C}|=m^{lt-n+k}$.
\end{itemize}

Conversely, assume that for any ideal $I$ of $\mathbb{P}$ with full count and $x\in\mathbb{Z}_m^n$, we have $(\ref{duality})$. Let $I\in\mathcal {I}^{\frac{n-k}{t}\cdot\left\lfloor\frac{m}{2}\right\rfloor}(\mathbb{P})$ be an ideal of $\mathbb{P}$ with full count and $u\in\mathbb{Z}_m^n$. By (\ref{duality}), we have $\left|B_I(u)\cap \mathcal {C}\right|=1$ which yields that $\mathcal {C}$ is $I$-perfect by Lemma \ref{IPERFECT}. Therefore, $\mathcal {C}$ is $I$-perfect for all $I\in\mathcal {I}^{\frac{n-k}{t}\cdot\left\lfloor\frac{m}{2}\right\rfloor}$  with full count. It follows from Theorem \ref{MDSSN} that $\mathcal {C}$ is an MDS $(\mathbb{P},\pi)$-code.
\end{proof}

\begin{example}
Let $\mathbb{P}=(M,R)$ be a pomset where $M=\{2/1,2/2,2/3\}$ and $$R=\{(4/(2/1,2/1),4/(2/2,2/2),4/(2/3,2/3),4/(2/1,2/2),4/(2/3,2/2)\}.$$
Let $\pi$ be a labeling of the pomset $\mathbb{P}$ such that $\pi(1)=\pi(2)=\pi(3)=2$. Consider the $(\mathbb{P},\pi)$-code $\mathcal {C}\subseteq\mathbb{Z}_5^6$ generated by the following matrix:
$$\left(
  \begin{array}{cccccc}
    0 & 1 & 1 & 2 & 2 & 3 \\
    1 & 0 & 0 & 2 & 2 & 3 \\
  \end{array}
\right)$$
Then the code $\mathcal {C}$ is a linear $(\mathbb{P},\pi)$-code over $\mathbb{Z}_5$ of length 6 with $d_{(\mathbb{P},\pi)}(\mathcal {C})=5$. We have $\left\lfloor\frac{d-1}{\left\lfloor\frac{m}{2}\right\rfloor}\right\rfloor=2$. The full count ideal in $\mathbb{P}$ such that $|I^*|=2$ is $I=\{2/1,2/3\}$. It follows that
$$\max\limits_{I\in\mathcal {I}(\mathbb{P}),|I^*|=2}\sum\limits_{i\in I^*}k_i=4=n-k.$$
Hence $\mathcal {C}$ is an MDS $(\mathbb{P},\pi)$-code. Moreover, $\mathcal {C}$ is $I$-perfect.
\end{example}

Let $\mathbb{P}=(M,R_1)$ and $\mathbb{Q}=(M,R_2)$ be two pomsets on $M$. We say that $\mathbb{Q}$ is \textbf{finer} than $\mathbb{P}$ if $p/a\ R_1\ q/b$ in $\mathbb{P}$ implies that $p/a\ R_2\ q/b$ in $\mathbb{Q}$.

\begin{lemma}
Let $\mathbb{P}_1$ and $\mathbb{P}_2$ be two pomsets on $M=\left\{\left\lfloor\frac{m}{2}\right\rfloor/1,\left\lfloor\frac{m}{2}\right\rfloor/2, \ldots,\left\lfloor\frac{m}{2}\right\rfloor/s\right\}$ and $\pi$ be a labeling of the pomset $\mathbb{P}_1$ with $k_1=k_2=\cdots=k_s=t$. If $\mathbb{P}_2$ is finer than $\mathbb{P}_1$, then every MDS $(\mathbb{P}_1,\pi)$-code is an MDS $(\mathbb{P}_2,\pi)$-code.
\end{lemma}

\begin{proof}
Suppose that $\mathcal {C}$ is an MDS $(n,K,d)$ $(\mathbb{P}_1,\pi)$-code. Then $n-\lceil\log_mK\rceil=t\cdot\left\lfloor \frac{d_{(\mathbb{P}_1,\pi)}(\mathcal {C})-1}{\left\lfloor\frac{m}{2}\right\rfloor}\right\rfloor$. Since $\mathbb{P}_2$ is finer than $\mathbb{P}_1$, we have $d_{(\mathbb{P}_1, \pi)}(u,v)\leq d_{(\mathbb{P}_2,\pi)}(u,v)$ for any $u,v\in\mathbb{Z}_m^n$. Therefore $d_{(\mathbb{P}_1,\pi)}(\mathcal {C})\leq d_{(\mathbb{P}_2,\pi)}(\mathcal {C})$ and hence
$$d_{(\mathbb{P}_2,\pi)}(\mathcal {C})-1\geq d_{(\mathbb{P}_1,\pi)}(\mathcal {C})-1\geq\frac{n-\lceil\log_mK\rceil}{t}\cdot \left\lfloor\frac{m}{2}\right\rfloor.$$
This implies that $\left\lfloor\frac{d_{(\mathbb{P}_2,\pi)}(\mathcal {C})-1}{\left\lfloor\frac{m}{2}\right\rfloor} \right\rfloor\geq\frac {n-\lceil\log_mK\rceil}{t}$. By the Singleton bound, $\mathcal {C}$ is an MDS $(\mathbb{P}_2,\pi)$-code.
\end{proof}

\begin{corollary}
An MDS block code is also an MDS pomset block code for every pomset defined on the set of blocks when all blocks have the same dimension.
\end{corollary}

\begin{proposition}
Let $(\mathbb{P},\pi)$ be a pomset block structure on $\mathbb{Z}_m^n$ with $k_1=k_2=\cdots=k_s=t$ and $\mathcal {C}$ be an $(n,m^k,d)$ $(\mathbb{P},\pi)$-code. If $\mathcal {C}$ is $(\frac{n-k}{t}\cdot\left\lfloor\frac{m}{2}\right\rfloor)$-perfect $(\mathbb{P},\pi)$-code, then $\mathcal {C}$ is an MDS $(\mathbb{P},\pi)$-code.
\end{proposition}

\begin{proof}
By Theorem \ref{MDSdual}, it is sufficient to show that $\mathcal {C}$ is an $I$-perfect code for all $I\in\mathcal {I}^{\frac{n-k}{t}\cdot\left\lfloor\frac{m}{2}\right\rfloor}$ with full count. Let $I$ be an ideal of cardinality $\frac{n-k}{t}\cdot\left\lfloor\frac{m}{2}\right\rfloor$ with full count. Suppose that there exist $c,c'\in \mathcal {C}$, $c\neq c'$ such that $B_I(c)=B_I(c')$. Then $\langle(supp_{(L,\pi)}(c-c')\rangle\subseteq I$. Since $\mathcal {C}$ is $(\frac{n-k}{t}\cdot\left\lfloor\frac{m}{2}\right\rfloor)$-perfect, the distance between two codewords of $\mathcal {C}$ is at least $\frac{n-k}{t}\cdot\left\lfloor\frac{m}{2}\right\rfloor+1$. Then
$$\frac{n-k}{t}\cdot\left\lfloor\frac{m}{2}\right\rfloor+1\leq d_{(\mathbb{P},\pi)}(c,c')\leq |I|=\frac{n-k}{t}\cdot\left\lfloor\frac{m}{2}\right\rfloor,$$
a contradiction.
\end{proof}

\begin{theorem}\label{NONLINEAR}
Let $(\mathbb{P},\pi)$ be a pomset block structure on $\mathbb{Z}_m^n$ with $k_1=k_2=\cdots=k_s=t$ and $\mathcal {C}$ be an $(n,K,d)$ $(\mathbb{P},\pi)$-code. If $\mathcal {C}$ is $I$-perfect for every ideal $I\in\mathcal {I}^{\frac{n-\lceil\log_mK\rceil}{t}\cdot \left\lfloor\frac{m}{2}\right\rfloor}(\mathbb{P})$, then $\mathcal {C}$ is an MDS $(\mathbb{P},\pi)$-code.
\end{theorem}

\begin{proof}
Assume that for any ideal $I\in\mathcal {I}^{\frac{n-\lceil\log_mK\rceil}{t}\cdot\left\lfloor\frac{m}{2}\right\rfloor}(\mathbb{P})$, $\mathcal {C}$ is $I$-perfect. To prove $\mathcal {C}$ is an MDS $(\mathbb{P},\pi)$-code, it is sufficient to show that $\left\lfloor\frac{d-1}{\left\lfloor\frac{m}{2}\right\rfloor}\right\rfloor\geq \frac{n-\lceil\log_mK\rceil}{t}$. Suppose that there exist two distinct codewords $u,v\in \mathcal {C}$ such that $d_{(\mathbb{P},\pi)}(u,v)\leq \frac{n-\lceil\log_mK\rceil}{t}\cdot\left\lfloor\frac{m}{2}\right\rfloor$. Denote by    $J=\langle supp_{(L,\pi)}(u-v)\rangle$. It follows form Proposition \ref{POMSET1} that there exists an ideal $I$ with cardinality $\frac{n-\lceil\log_mK\rceil}{t}\cdot\left\lfloor\frac{m}{2}\right\rfloor$ such that $J\subseteq I$. Then $u-v\in B_J\subseteq B_I$ which implies that $u\in B_I(v)$, a contradiction to the fact that $\mathcal {C}$ is $I$-perfect. Therefore $d>\frac{n-\lceil\log_mK\rceil}{t}\cdot\left\lfloor\frac{m}{2}\right\rfloor$ and hence $\left\lfloor\frac{d-1}{\left\lfloor\frac{m}{2}\right\rfloor}\right\rfloor \geq\frac{n-\lceil\log_mK\rceil}{t}$.
\end{proof}

\subsection{Weight distribution}

\quad\; In what follows, we consider a $(\mathbb{P},\pi)$-code $\mathcal {C}$ whose blocks have the same dimension $t$ and $\mathbb{P}$ is a chain pomset on $M=\left\{\left\lfloor\frac{m}{2}\right\rfloor/1,\left\lfloor\frac{m}{2}\right\rfloor/2, \ldots,\left\lfloor\frac{m}{2}\right\rfloor/s\right\}$ such that $\left\lfloor\frac{m}{2}\right\rfloor/1\ R\ \left\lfloor\frac{m}{2}\right\rfloor/2\ R\ \ldots R\left\lfloor\frac{m}{2}\right\rfloor/s$.

Note that every ideal $I$ in the chain pomset $\mathbb{P}$ has the form $I=\left\{\left\lfloor\frac{m}{2}\right\rfloor/1,\left\lfloor\frac{m}{2}\right\rfloor/2, \ldots,\left\lfloor\frac{m}{2}\right\rfloor/(i-1),p/i\right\}$ where $1\leq p\leq \left\lfloor\frac{m}{2}\right\rfloor$ and $p/i$ is the unique maximal element in $I$. Moreover, for any integer $1\leq r\leq s\cdot\left\lfloor\frac{m}{2}\right\rfloor$, there is only one ideal $I$ in $\mathbb{P}$ whose cardinality is $r$, that is, $\left|\mathcal {I}^r(\mathbb{P})\right|=1$ and hence $B_r(u)=B_I(u)$ for any $u\in\mathbb{Z}_m^n$.

\begin{theorem}\label{CHAIN}
Let $(\mathbb{P},\pi)$ be a chain pomset block structure on $\mathbb{Z}_m^n$ with $k_1=k_2=\cdots=k_s=t$. Then for any $K$ such that $\lceil\log_mK\rceil$ is divisible by $t$, every $I$-perfect $(n,K,d)$ $(\mathbb{P},\pi)$-code $\mathcal {C}$ is an MDS $(\mathbb{P},\pi)$-code.
\end{theorem}

\begin{proof}
Let $\mathcal {C}\subseteq\mathbb{Z}_m^n$ be an $I$-perfect $(n,K,d)$ $(\mathbb{P},\pi)$-code. Then $\bigsqcup\limits_{c\in \mathcal {C}}B_I(c)=\mathbb{Z}_m^n$ and $d>|I|$. From this, we have $|\mathcal {C}|\cdot|B_I|=m^n$.
\begin{itemize}
  \item {\bf Case 1:}If $I$ is an ideal with full count, then $|I^{*}|=\frac{n-\log_mK}{t}$ and $|I|=\frac{n-\log_mK}{t}\cdot\left\lfloor\frac{m}{2}\right\rfloor$. It follows from Theorem \ref{NONLINEAR} that $\mathcal {C}$ is an MDS $(\mathbb{P},\pi)$-code.
  \item {\bf Case 2:}  If $I$ is an ideal with partial count, then $|I|=|J|+l$ for some ideal $J$ with full count and $1\leq l\leq\left\lfloor\frac{m}{2}\right\rfloor-1$. Since $\mathcal {C}$ is $I$-perfect, we have
      $$|\mathcal {C}|\cdot|B_I|=|\mathcal {C}|\cdot(2l+1)^t\cdot|B_J|=K\cdot(2l+1)^t\cdot m^{|J^*|\cdot t}=m^n,$$
      yielding to
      $$\left|J^*\right|=\frac{n}{t}-\frac{\log_mK}{t}-\log_m(2l+1).$$
      Since $1\leq l\leq \left\lfloor\frac{m}{2}\right\rfloor-1$, we have $0<\log_m(2l+1)<1$. Note that $\frac{\log_mK}{t}+\log_m(2l+1)\in \mathbb{Z}$. Therefore, $$\frac{\log_mK}{t}+\log_m(2l+1)=\left\lfloor\frac{\log_mK}{t}+ \log_m(2l+1)\right\rfloor\leq \left\lfloor\frac{\log_mK}{t}\right\rfloor+ \left\lfloor\log_m(2l+1)\right\rfloor+1= \left\lceil\frac{\log_mK}{t}\right\rceil.$$
      Hence $|J^*|\geq \frac{n}{t}-\left\lceil\frac{\log_mK}{t}\right\rceil$. On the other hand, $d>|I|=|J|+l$ implies that $\frac{d-1}{\left\lfloor\frac{m}{2}\right\rfloor}>|J^{*}|+ \frac{l-1}{\left\lfloor\frac{m}{2}\right\rfloor}$. Thus,
      $$\left\lfloor\frac{d-1}{\left\lfloor\frac{m}{2}\right\rfloor}\right\rfloor\geq \left|J^*\right|\geq\frac{n}{t}-\left\lceil\frac{\log_mK}{t}\right\rceil.$$
      Since $\lceil\log_mK\rceil$ is divisible by $t$, we have
      $$t\cdot\left\lfloor\frac{d-1}{\left\lfloor\frac{m}{2}\right\rfloor}\right\rfloor \geq n-\left\lceil\log_mK\right\rceil.$$
      By the Singleton bound,  $\mathcal {C}$ is an MDS code.
\end{itemize}
\end{proof}

\begin{remark}
The above theorem holds for any ideal $I\in \mathcal {I}(\mathbb{P})$ no matter that $I$ has full count or partial count.
\end{remark}

Let $I$ be an ideal of $\mathbb{P}$. We denote by $I_p$ ($I_f$ resp.) the collection of the elements in $I^*$ which has partial count (full count resp.) in $I$. From the proof of Theorem \ref{CHAIN}, we obtain a corollary.

\begin{corollary}\label{general}
Let $(\mathbb{P},\pi)$ be a pomset block structure on $\mathbb{Z}_m^n$ with $k_1=k_2=\cdots=k_s=t$. Let $I$ be an ideal of $\mathbb{P}$ with partial count. Suppose that  $I_p=\{l_1,l_2,\ldots,l_{\lambda}\}\subseteq I^{*}$ and $1<\sum\limits_{i=1}^{\lambda}\log_m(2C_I(l_i)+1)<1$. Let $\mathcal {C}$ be an $(n,K,d)$ $(\mathbb{P},\pi)$-code such that $\lceil\log_mK\rceil$ is divisible by $t$ and $\left\lfloor\frac{d-1}{\left\lfloor\frac{m}{2}\right\rfloor}\right\rfloor\geq |I_f|$. If $\mathcal {C}$ is $I$-perfect, then $\mathcal {C}$ is MDS.
\end{corollary}

Let $I\in \mathcal {I}(\mathbb{P})$ be an ideal with partial count.  Here is an example illustrates that an $I$-perfect $(\mathbb{P},\pi)$-code is not necessarily an MDS code. But when we give some restrictions on $I$, an $I$-perfect code could be an MDS code.

\begin{example}
Let $\mathbb{P}=(M,R)$ be a pomset where $M=\{4/1,4/2\}$ and $$R=\{(16/(4/1,4/1),16/(4/2,4/2)\}.$$
Let $\pi$ be a labeling of the pomset $\mathbb{P}$ such that $\pi(1)=\pi(2)=1$. Consider the $(\mathbb{P},\pi)$-code $\mathcal {C}\subseteq\mathbb{Z}_9^2$ defined by $\mathcal {C}=\{(0,0),(0,3),(0,6)\}$. Then $d_{(\mathbb{P},\pi)}(\mathcal {C})=3$ and $\left\lfloor\frac{d_{(\mathbb{P},\pi)}(\mathcal {C})-1}{\left\lfloor\frac{m}{2}\right\rfloor}\right\rfloor=0$. Therefore $\mathcal {C}$ is not MDS. Let $I=\{4/1,1/2\}$ be an ideal of $\mathbb{P}$ with partial count. It can be easily seen that $\mathcal {C}$ is $I$-perfect.

If we define $\mathcal {C}'=\{(0,0),(2,3),(4,6)\}$. It is routine to verify that $\mathcal {C}'$ is an $I$-perfect code. By Corollary \ref{general}, $\mathcal {C}'$ is an MDS $(\mathbb{P},\pi)$-code. In fact, $d_{(\mathbb{P},\pi)}(\mathcal {C}')=5$ and $\left\lfloor\frac{d_{(\mathbb{P},\pi)}(\mathcal {C}')-1}{\left\lfloor\frac{m}{2}\right\rfloor}\right\rfloor=1$. The ideal of $\mathbb{P}$ such that $|I^*|=1$ has root set $\{1\}$ or $\{2\}$. It follows from $k_1=k_2=1$ that $\mathcal {C}'$ is an MDS $(\mathbb{P},\pi)$-code.

Furthermore, when $\mathbb{P}$ is chain pomset such that $4/1\ R\ 4/2$, we can see that both $\mathcal {C}$ and $\mathcal {C}'$ are MDS $(\mathbb{P},\pi)$-codes.
\end{example}

Denote by $A_{r,(\mathbb{P},\pi)}(\mathcal {C})$ the number of codewords of $(\mathbb{P},\pi)$-weight $r$ in $\mathcal {C}$, that is,
$$A_{r,(\mathbb{P},\pi)}(\mathcal {C})=\left|\left\{c\in \mathcal {C}: w_{(\mathbb{P},\pi)}(c)=r\right\}\right|.$$

Note that an ideal $I$ of $\mathbb{P}$ is full count if and only if $|I|=l\cdot\left\lfloor\frac{m}{2}\right\rfloor$.

\begin{lemma}\label{BIC}
Let $(\mathbb{P},\pi)$ be a pomset block structure on $\mathbb{Z}_m^n$ where $\mathbb{P}$ is the chain pomset on $M=\left\{\left\lfloor\frac{m}{2}\right\rfloor/1,\left\lfloor\frac{m}{2}\right\rfloor/2, \ldots,\left\lfloor\frac{m}{2}\right\rfloor/s\right\}$ and $\pi$ is a labeling of the pomset $\mathbb{P}$ with $k_1=k_2=\cdots=k_s=t$. Let $\mathcal {C}$ be an MDS linear $(n,m^k,d)$ $(\mathbb{P},\pi)$-code and let $I\in\mathcal {I}(\mathbb{P})$ be an ideal of $\mathbb{P}$. Then
$$|B_I\cap \mathcal {C}|=\left\{
                \begin{array}{ll}
                 1, & \text{if}\ |I|\leq\frac{(n-k)}{t}\cdot\left\lfloor\frac{m}{2}\right\rfloor,\\[2mm]
                m^{tl-n+k}, & \text{if}\ |I|=l\cdot\left\lfloor\frac{m}{2}\right\rfloor\geq \frac{(n-k)}{t}\cdot\left\lfloor\frac{m}{2}\right\rfloor,\\[2mm]
                (2p+1)^tm^{tl-n+k}, & \text{if}\ |I|=l\cdot\left\lfloor\frac{m}{2}\right\rfloor+p> \frac{(n-k)}{t}\cdot\left\lfloor\frac{m}{2}\right\rfloor,\ 1\leq p\leq \left\lfloor\frac{m}{2}\right\rfloor-1.
                \end{array}
              \right.$$
\end{lemma}

\begin{proof}
Let $I$ be an ideal with cardinality $l\cdot\left\lfloor\frac{m}{2}\right\rfloor+p>\frac{(n-k)}{t}\cdot \left\lfloor\frac{m}{2}\right\rfloor$ where $1\leq p\leq\left\lfloor\frac{m}{2}\right\rfloor-1$. Then $I$ is partial count with $|I^{*}|=l+1$. By Proposition \ref{pomset1}, there exists an ideal $J$ with cardinality $\frac{n-k}{t}\cdot\left\lfloor\frac{m}{2}\right\rfloor$ such that $J\subseteq I$. It follows from Theorem \ref{MDSdual} that $\mathcal {C}$ is $J$-perfect and $|B_J(x)\cap \mathcal {C}|=1$ for any $x\in\mathbb{Z}_m^n$. Let $K=I\ominus J$. The translates $x+B_J$, $x\in B_K$ are disjoint and their union covers $B_I$. Therefore,
$$|B_I\cap \mathcal {C}|=|B_K|=(2p+1)^tm^{(|I^{*}|-|J^{*}|-1)t}=(2p+1)^tm^{(l+1-\frac{n-k}{t}-1)t}=(2p+1)^t m^{tl-n+k}.$$
This completes the proof.
\end{proof}

\begin{theorem}\label{weight}
Let $(\mathbb{P},\pi)$ be a chain pomset block structure with $k_1=k_2=\cdots=k_s=t$ and $\mathcal {C}$ be an MDS linear $(n,m^k,d)$ $(\mathbb{P},\pi)$-code. Then
$$A_{r,(\mathbb{P},\pi)}(\mathcal {C})=\left\{
                              \begin{array}{ll}
                                1, & \text{if}\ r=0, \\[2mm]
                                0, & \text{if}\ 1\leq r\leq d-1,\\[2mm]
   \left(m^t-\left(2\left\lfloor\frac{m}{2}\right\rfloor-1\right)^t\right)m^{tl-n+k-t}, & \text{if}\ r=l\cdot\left\lfloor\frac{m}{2}\right\rfloor\geq d,\\[2mm]
                (3^t-1)m^{tl-n+k}, & \text{if}\ r=l\cdot\left\lfloor\frac{m}{2}\right\rfloor+1\geq d,\\[2mm]
   \left((2p+1)^t-(2p-1)^t\right)m^{tl-n+k}, & \text{if}\ r=l\cdot\left\lfloor\frac{m}{2}\right\rfloor+p\geq d\ \text{and}\ 2\leq p\leq\left\lfloor\frac{m}{2}\right\rfloor-1.
                           \end{array}
                            \right.
$$
\end{theorem}

\begin{proof}
If $r\leq d-1$, then the result is trivial. So we assume that $r\geq d$. Suppose that ${I}^r(\mathbb{P})=\{I\}$ and $I^{r-1}(\mathbb{P})=\{J\}$. Then
$$A_{r,(\mathbb{P},\pi)}(\mathcal {C})=|B_r\cap \mathcal {C}|-|B_{r-1}\cap \mathcal {C}|=|B_I\cap \mathcal {C}|-|B_J\cap \mathcal {C}|.$$
\begin{itemize}
  \item {\bf Case 1:} Assume that $r=l\cdot\left\lfloor\frac{m}{2}\right\rfloor$. Then $I$ has full count, $J$ has partial count and $|I^{*}|=|J^{*}|=l$. By Lemma \ref{BIC},
$$|B_I\cap \mathcal {C}|-|B_J\cap \mathcal {C}|=m^{tl-n+k}-\left(2\left\lfloor\frac{m}{2}\right\rfloor-1\right)^tm^{tl-n+k-t}= m^{tl-n+k-t} \left(m^t-\left(2\left\lfloor\frac{m}{2}\right\rfloor-1\right)^t\right).$$
  \item {\bf Case 2:} Assume that $r=l\cdot\left\lfloor\frac{m}{2}\right\rfloor+1$. Then $I$ has partial count, $J$ has full count and $|J^{*}|=|I^{*}|-1=l$. By Lemma \ref{BIC},
$$\left|B_I\cap \mathcal {C}\right|-\left|B_J\cap \mathcal {C}\right|=3^tm^{tl-n+k}-m^{tl-n+k}=(3^t-1)m^{tl-n+k}.$$
  \item {\bf Case 3:} Assume that $r=l\cdot\left\lfloor\frac{m}{2}\right\rfloor+p$ where $2\leq p\leq\left\lfloor\frac{m}{2}\right\rfloor-1$. The $I$ and $J$ have partial count and $|I^{*}|=|J^{*}|=l+1$. By Lemma \ref{BIC},
$$|B_I\cap \mathcal {C}|-|B_J\cap \mathcal {C}|=(2p+1)^tm^{tl-n+k}-(2p-1)^tm^{tl-n+k}=\left((2p+1)^t-(2p-1)^t\right)m^{tl-n+k}.$$
\end{itemize}
\end{proof}

\begin{remark}
Note that when pomset block metric is defined on $\mathbb{Z}_2$ and $\mathbb{Z}_3$, pomset block and poset block weights coincide. The weight distribution of an MDS $(P,\pi)$-code for the case when all blocks have the same dimension is determined (see [\ref{DASS}, Theorem 5.2]). Consider two special cases of hierarchical pomset. When the hierarchical pomset is a chain pomset, Theorem \ref{weight} gives the weight distribution of an  MDS $(\mathbb{P},\pi)$-code for the case when all blocks have the same dimension. When the hierarchical pomset is an antichain pomset and all blocks have dimension 1, pomset block weight is the traditional Lee weight. As our best knowledge, the Lee weight distribution of an MDS code has not been obtained for general case.
\end{remark}

\begin{remark}
When $t=1$, Case 2 and Case 3 in Theorem \ref{weight} would coincide which is exactly Theorem 9 in [\ref{POMSET2}].
\end{remark}

\section{Conclusion and Further Consideration}

\quad\;The  pomset metric is a generalization of poset metric and gives rise to Lee metric if the underlying pomset is an antichain. In this paper, we introduce MDS pomset block codes and extend the concept of $I$-perfect codes to the case of pomset block metric.

After the introduction of pomset block codes, it would be interesting to construct Macwilliams type identities for any linear code with chain block pomset and bound for covering radius of product codes. By carefully checking relevant results on pomset codes and poset block codes, we may have a chance to explore further properties of pomset block codes.

\end{document}